\documentclass[acmsmall,screen]{acmart}

\usepackage{wrapfig}
\usepackage{tikz}

\newcommand{\figlabel}[1]{\label{fig:#1}}
\newcommand{\figref}[1]{Fig.~\ref{fig:#1}}

\newcommand{\seclabel}[1]{\label{sec:#1}}
\newcommand{\secref}[1]{Section~\ref{sec:#1}}
\newcommand{\exlabel}[1]{\label{ex:#1}}
\newcommand{\exref}[1]{Example~\ref{ex:#1}}
\newcommand{\deflabel}[1]{\label{def:#1}}
\newcommand{\defref}[1]{Definition~\ref{def:#1}}

\newcommand{\thmlabel}[1]{\label{thm:#1}}

\newcommand{\lemlabel}[1]{\label{lem:#1}}
\newcommand{\lemref}[1]{Lemma~\ref{lem:#1}}

\newcommand{\remlabel}[1]{\label{rem:#1}}
\newcommand{\remref}[1]{Remark~\ref{rem:#1}}

\newcommand{\applabel}[1]{\label{app:#1}}

\theoremstyle{definition}
\newtheorem{example}{Example}
\numberwithin{example}{section}

\newtheorem{theorem}{Theorem}
\numberwithin{theorem}{section}

\newtheorem{definition}{Definition}
\numberwithin{definition}{section}

\numberwithin{problem}{section}

\newtheorem{lemma}{Lemma}
\numberwithin{lemma}{section}

\numberwithin{proposition}{section}

\numberwithin{corollary}{section}

\numberwithin{claim}{section}

\newtheorem{remark}{Remark}

\newenvironment{myparagraph}[1]{\vspace{0.1in}\noindent{\bf #1.}}{}

\DeclareMathAlphabet{\mathpzc}{OT1}{pzc}{m}{it}

\newcommand{\nats}{\mathbb{N}}

\renewcommand{\emptyset}{\varnothing}

\newcommand{\protocol}{\mathcal{P}}

\newcommand{\ntypes}{\mathcal{N}}          
\newcommand{\mtypes}{\mathcal{M}}          
\newcommand{\lstate}{L}                     
\newcommand{\gstate}{G}                     
\newcommand{\msgs}{\mathtt{Msg}}
\newcommand{\nodes}{\mathtt{N}}               
\newcommand{\acts}{\mathcal{A}}
\newcommand{\exec}{\mathit{\Pi}}
\newcommand{\fexec}{\exec^{\text{feas}}}
\newcommand{\id}{\mathtt{id}}

\newcommand{\instances}{\mathcal{I}}
\newcommand{\thy}{\mathbb{T}}                
\newcommand{\interp}{\mathcal{J}}            

\newcommand{\domain}{\mathcal{D}}
\newcommand{\sender}{\mathtt{sender}}
\newcommand{\receiver}{\mathtt{receiver}}
\newcommand{\payload}{\mathtt{payload}}
\newcommand{\ntype}{\mathtt{ntype}}        
\newcommand{\mtype}{\mathtt{mtype}}       
\newcommand{\gd}{\mathtt{gd}}             
\newcommand{\eff}{\mathtt{eff}}

\newcommand{\reach}{\mathtt{reach}}
\newcommand{\sem}[1]{\llbracket #1 \rrbracket}

\newcommand{\actsI}{\acts_I}                 
\newcommand{\runl}{\mathtt{run}}             
\newcommand{\rew}{\rightsquigarrow}          

\newcommand{\scvars}{V}                       
\newcommand{\dom}{\mathcal{D}}                
\newcommand{\assm}{\mathtt{assume}}           
\newcommand{\evalx}[2]{\sem{#1}_{#2}}   

\newcommand{\semSig}[1]{\sem{#1}_{\mathsf{struct}}}    
\newcommand{\semvar}[1]{\sem{#1}_{\mathsf{var}}}         
\newcommand{\semfeas}[1]{\sem{#1}_{\mathsf{feas}}} 
\newcommand{\semcom}[1]{\sem{#1}_{\mathsf{com}}}    

\newcommand{\incSig}{\sqsubseteq_{\mathsf{struct}}}
\newcommand{\incvar}{\sqsubseteq_{\mathsf{var}}}
\newcommand{\incfeas}{\sqsubseteq_{\mathsf{feas}}}
\newcommand{\inccom}{\sqsubseteq_{\mathsf{com}}}
\newcommand{\inctop}{\sqsubseteq}
\newcommand{\inclayer}{\sqsubseteq_{\mathsf{layer}}}

\newcommand{\rulename}[1]{\textsc{#1}}

\RequirePackage{alltt}
  {\par\noindent\begin{minipage}{\linewidth}\small\begin{alltt}}%
  {\end{alltt}\end{minipage}\par}

\newcommand{\tT}{\ensuremath{\mathsf{true}}}
\newcommand{\tF}{\ensuremath{\mathsf{false}}}

\usepackage{amsmath,amsthm,amsbsy}
\usepackage{graphicx}
\usepackage{caption,wrapfig}
\usepackage{subcaption}
\usepackage{longtable}
\usepackage{booktabs}
\usepackage{tabu}
\usepackage{tabularx}
\usepackage[referable]{threeparttablex}
\usepackage{multirow}
\usepackage{subcaption}
\usepackage{hyperref}
\usepackage[capitalise]{cleveref}
\usepackage{stmaryrd}
\usepackage{mathtools} 
\usepackage{makecell}
\crefname{figure}{Figure}{Figure}
\crefformat{footnote}{#2\footnotemark[#1]#3}
\usepackage{tikz}
\usepackage{comment}
\usepackage{parskip}
\usepackage{paralist}
\usepackage{pgfplots}

\usepackage[ruled, linesnumbered, noend]{algorithm2e}
\usepackage{tikz}
\usetikzlibrary{shapes.misc,arrows.meta}
\usetikzlibrary{calc,decorations.pathmorphing,shapes}
\usetikzlibrary{decorations.markings}
\usetikzlibrary{decorations.pathreplacing}
\tikzset{snake it/.style={decorate, decoration=snake}}

\usepackage{xspace}
\usepackage{xcolor}
\usepackage[inline, shortlabels]{enumitem}

\usepackage{multicol}
\usepackage{bbm}

\usepackage{thm-restate}

\begin{document}

\title{Scenario Proofs for Verification of Distributed Protocols}


\begin{abstract}
Formal verification of distributed protocols is still dominated by inductive-invariant proofs, whose central challenge is discovering strengthening assertions that rule out unreachable states. This burden is often misaligned with how protocol designers reason: they explain correctness by following representative executions, causal histories, and phase-structured interactions among nodes. This paper introduces \emph{scenario proofs}, a verification methodology that makes such execution-oriented reasoning formal.

A scenario is a structured expression describing representative protocol executions, enriched with ghost variables, guards, sequencing, choice, and iteration. Rather than enumerate all possible interleavings, a scenario describes canonical executions and uses a sound semi-commutation relation to account for reorderings that preserve feasibility and final states. Verification then factors into two main obligations. \emph{Coverage} shows that every feasible protocol execution is equivalent, up to semi-commutation, to one described by the scenario; this obligation is independent of the safety property and serves as a reusable behavioral characterization of the protocol. \emph{Correctness} proves the desired safety property only over the representative executions, using Hoare-style weakest-precondition reasoning along the scenario structure.

We formalize the scenario language, its layered semantics, and proof rules for coverage and correctness. Through case studies of ring leader election, two-phase commit, a decentralized lock, two variants of multi-signature payment, and single-decree Paxos, we show how scenario proofs capture the causal and phase-based arguments used informally by protocol designers. These examples illustrate that scenario proofs can replace monolithic invariant discovery with structured reasoning over representative executions, while retaining a rigorous connection to all feasible protocol behaviors.
\end{abstract}


\author{Zhendong Ang}
\affiliation{%
  \institution{National University of Singapore}
  \country{Singapore}
}
\email{zhendong.ang@u.nus.edu.sg}

\author{George P\^{i}rlea}
\affiliation{%
  \institution{National University of Singapore}
  \country{Singapore}
}
\email{gpirlea@u.nus.edu}

\author{Umang Mathur}
\affiliation{%
  \institution{National University of Singapore}
  \country{Singapore}
}
\email{umathur@nus.edu.sg}

\maketitle

\section{Introduction}
\seclabel{intro}
Distributed protocols form the backbone of modern computing systems, from cloud infrastructure and
database replication to blockchain networks and distributed file systems.
Ensuring their correctness is crucial, and formal verification of distributed protocols
has gained significant attention in both research and industry.
Yet formally proving the correctness of distributed protocols remains remarkably difficult.
The source of this difficulty is fundamental:
correctness of distributed protocols is typically expressed as \emph{safety properties}
--- guarantees that certain bad states
(e.g., two nodes committing different values, two nodes believing they are the leader simultaneously)
are never reached by \emph{all} possible executions of the protocol,
i.e., all possible interleavings of concurrent actions by multiple nodes.

The predominant approach to this challenge involves strengthening the safety property into an
\emph{inductive invariant}, a property that
\begin{enumerate*}
    \item holds in the initial state of the protocol,
    \item is preserved by every action of the protocol, and
    \item implies the desired safety property.
\end{enumerate*}
If such an inductive invariant is identified and proven,
then one can inductively conclude that the invariant holds in all reachable states,
and thus the safety property holds as well.
This approach has been successfully employed in numerous verification frameworks and tools~\cite{ivy, ironfleet, tlaplus, verdi, veil}.

However, devising the right inductive invariant is the central difficulty in this approach,
and we argue that it is \emph{fundamentally misaligned} with how protocol designers actually think.
The core issue is that inductive invariant reasoning isolates actions and states from their execution context that reaches them.
Recall that an inductive invariant $\mathtt{inv}$ is correct if,
for any action $a$, starting from \emph{any} state $s$ that satisfies $\mathtt{inv}$,
the successor state $s'$ reached after executing $a$ also satisfies $\mathtt{inv}$.
The quantification ranges over \emph{all} states satisfying $\mathtt{inv}$,
not just those that are actually reachable through valid protocol executions.
To exclude all such unreachable states $s$ that prevent the candidate invariant from being inductive,
an iterative strengthening loop is performed by carefully identifying invariant conjuncts and adding them to the candidate invariant.
This process is the primary bottleneck in protocol verification efforts and 
reasoning about the unreachable states is absent from the \emph{natural reasoning} by the protocol designers.
Protocol designers naturally reason about correctness by examining the executions
that \emph{could} happen: they consider sequences of actions
and the reachable states these sequences lead to.
To prove the safety of the reachable states,
the sequences of actions leading to them usually carry the needed contextual information:
what actions have occurred, in what order, and thus, why certain dangerous situations cannot arise.

Indeed, when introducing a distributed protocol in textbooks, lectures, or technical documentation, 
practitioners argue about correctness by presenting representative executions.
They walk through example cases showing how nodes interact
and explain why these interactions maintain safety.
Even when devising inductive invariants manually,
verification engineers often imagine specific executions first to guide their thinking,
and then extract invariants from these execution patterns.
Our experience learning and verifying distributed protocols confirms
that execution-oriented reasoning is the natural way humans understand protocol correctness.
We call these structured descriptions of representative executions \emph{scenarios}:
they capture the key logical patterns underlying protocol behavior and form
the conceptual building blocks through which protocols are designed and understood.
However, such scenario-based reasoning remains informal.

In this paper, inspired by our own experience and a recent work on scenario-based proofs
for concurrent objects~\cite{enea-koskinen},
we aim to bridge the gap between this natural reasoning and formal verification.
We propose \emph{scenario-based proofs} for distributed protocols, 
a verification methodology that allows protocol designers
to formally describe their execution-oriented reasoning as \emph{scenarios}
and prove protocol correctness directly from these scenarios.
Our framework captures the natural way developers reason about protocols
and gives this reasoning the formal rigor of a proof methodology.

\myparagraph{Motivating Example}
To make the above discussion concrete, consider the ring leader election protocol
as a running example (\figref{rle-protocol}).
In this protocol, nodes are arranged in a ring topology,
and each node initiates a vote for itself by sending its unique ID to the next node in the ring.
When a node receives a vote with ID $v$, it compares $v$ with its own ID:
if $v$ is larger, the node forwards $v$ to the next node; otherwise, it discards the vote.
A node is elected as leader when it receives its own ID back,
indicating that its ID has traversed the entire ring without being discarded.
The safety property states that at most one node can be elected as leader.

\textit{The inductive invariant approach and its difficulties.}\;
A typical inductive invariant for this protocol is a conjunction of the following:
\begin{enumerate}[nosep]
    \item[(I1)] If a node is elected as leader, it has the maximum ID among all nodes.
    \item[(I2)] If a node receives a vote with its own ID, it has the maximum ID.
    \item[(I3)] If a vote with ID $v$ is in transit to node $n$, then all nodes between $v$'s
    originator and $n$ (in the ring order) have IDs less than $v$.
\end{enumerate}
Let us consider the process of \emph{arriving} at this invariant.
Suppose we already have (I1) and (I2) and attempt to show their conjunction is inductive.
A counterexample arises: a state where a vote with $\mathtt{id}(n)$ is in transit towards node $n$,
whose ID is not the maximum, satisfies (I1) $\land$ (I2).
Then the successor state, where $n$ receives its own ID, violates I2.
To exclude this unreachable state, the engineer must add a new conjunct,
but there is no intuitive guide for \emph{what} to add.
One might try: ``if a vote with $\mathtt{id}(n)$ is in transit to $n$, then $n$ has the maximum ID.''
This is true but still not inductive.
The counterexample is one step backward:
a vote $\mathtt{id}(n)$ that is in transit to $m$, the previous node of $n$ and $\mathtt{id}(m) < \mathtt{id}(n)$, will trigger the same vote sent to $n$ in the successor state,
which violates the newly added conjunct.
A stronger attempt might assert that such in-transit votes never exist at all,
but this is simply false.
The correct strengthening is (I3), which generalizes the first attempt and, tracks the relationship between an in-transit vote and every node it has passed.
Arriving at (I3) requires intellectual effort with no clear guidance from the counterexample alone.

\begin{figure}[t]
\small
\hrule\smallskip
\begin{minipage}[t]{0.48\linewidth}
\textbf{Protocol: Ring Leader Election}

\bigskip
\textit{Topology:} nodes ordered in a ring
\begin{itemize}[leftmargin=1.2em]
    \item $\mathtt{next}(i)$ = next node in the ring
    \item $\mathtt{id}(i)$ = unique identifier of node $i$
\end{itemize}

\bigskip
\textit{Local State} (of Node $i$)
\begin{itemize}[leftmargin=1.2em]
    \item $\mathtt{isInited}, \mathtt{isElected}: \mathsf{Bool}$, both initially $\mathsf{False}$.
\end{itemize}

\bigskip
\textit{Global State}
\begin{itemize}[leftmargin=1.2em]
    \item $\mathtt{sent}, \mathtt{recv}$: message buffers, initially empty.
\end{itemize}

\bigskip
\textit{Action}
\begin{itemize}[leftmargin=1.2em]
    \item $\mathtt{Init}$ by Node $i$:\\
    \hspace*{1em}\textit{require:}\ $\neg\, i.\mathtt{isInited}$\\
    \hspace*{1em}send $\mathtt{id}(i)$ to $\mathtt{next}(i)$;\quad $i.\mathtt{isInited} \leftarrow \mathsf{True}$
    \item $\mathtt{Recv}$ by Node $i$ on a vote $v$:\\
    \hspace*{1em}\textit{require:}\ $v \in \mathtt{sent} \,\wedge\, v \notin \mathtt{recv} \,\wedge\, v.\mathtt{receiver} = i$\\
    \hspace*{1em}add $v$ to $\mathtt{recv}$;\\
    \hspace*{1em}$\mathbf{if}\; v > \mathtt{id}(i){:}\ \text{send }v \text{ to }\mathtt{next}(i)$\\
    \hspace*{1em}$\mathbf{if}\; v = \mathtt{id}(i){:}\ i.\mathtt{isElected} \leftarrow \mathsf{True}$
\end{itemize}

\end{minipage}
\hfill\vrule\hfill
\begin{minipage}[t]{0.48\linewidth}
\centering

\textbf{Example topology}
\smallskip

\begin{tikzpicture}[
    every node/.style={circle, draw, minimum size=0.9cm, font=\normalsize},
    >=Stealth,
]
    \node (n1) at (90:1.3cm)  {$1$};
    \node (n3) at (210:1.3cm) {$3$};
    \node (n2) at (330:1.3cm) {$2$};

    \draw[->, thick] (n1) to[bend right=20] (n3);
    \draw[->, thick] (n3) to[bend right=20] (n2);
    \draw[->, thick] (n2) to[bend right=20] (n1);
\end{tikzpicture}

\medskip
\textbf{Causal histories}
\smallskip

\begin{tikzpicture}[
    cnode/.style={circle, draw, minimum size=0.6cm, font=\small, inner sep=0pt},
    msg/.style={->, thick, >=Stealth},
    elect/.style={->, thick, >=Stealth, green!50!black},
    chainlabel/.style={font=\small\itshape, anchor=east},
]
    \node[chainlabel] at (-0.2, 0) {$1$:};
    \node[cnode] (a1) at (0.4, 0) {$1$};
    \node[cnode] (a2) at (1.6, 0) {$3$};
    \draw[msg] (a1) -- node[above, font=\scriptsize] {$1$} (a2);
    \node[font=\scriptsize, red!70!black, anchor=west] at (2.0, 0) {$\times$};

    \node[chainlabel] at (-0.2, -0.9) {$2$:};
    \node[cnode] (b1) at (0.4, -0.9) {$2$};
    \node[cnode] (b2) at (1.6, -0.9) {$1$};
    \node[cnode] (b3) at (2.8, -0.9) {$3$};
    \draw[msg] (b1) -- node[above, font=\scriptsize] {$2$} (b2);
    \draw[msg] (b2) -- node[above, font=\scriptsize] {$2$} (b3);
    \node[font=\scriptsize, red!70!black, anchor=west] at (3.2, -0.9) {$\times$};

    \node[chainlabel] at (-0.2, -1.8) {$3$:};
    \node[cnode] (c1) at (0.4, -1.8) {$3$};
    \node[cnode] (c2) at (1.6, -1.8) {$2$};
    \node[cnode] (c3) at (2.8, -1.8) {$1$};
    \node[cnode, double, green!50!black] (c4) at (4.0, -1.8) {$3$};
    \draw[msg] (c1) -- node[above, font=\scriptsize] {$3$} (c2);
    \draw[msg] (c2) -- node[above, font=\scriptsize] {$3$} (c3);
    \draw[elect] (c3) -- node[above, font=\scriptsize] {$3$} (c4);
    \node[font=\scriptsize, green!50!black, anchor=west] at (4.4, -1.8) {\checkmark};
\end{tikzpicture}

\end{minipage}
\smallskip\hrule
\caption{%
    Ring Leader Election protocol (left) and example instance (right).
    The ring has topology $1 \to 3 \to 2 \to 1$.
    Each causal history traces one vote's propagation:
    votes $1$ and $2$ are discarded ($\times$) upon reaching node $3$;
    vote $3$ completes the full ring and node $3$ is elected (\checkmark).%
}
\figlabel{rle-protocol}
\end{figure}
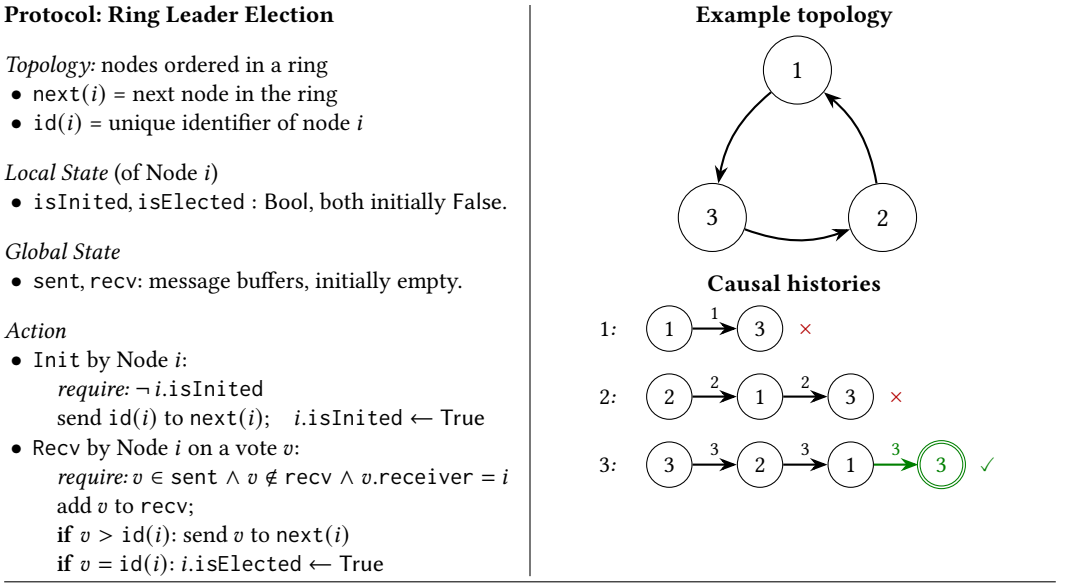

Interestingly, the final invariant (I1) $\land$ (I2) $\land$ (I3), once found,
does mirror a sequence of causally enabled actions:
(I1) captures the condition of being elected;
(I2) captures the final vote forwarding action that enables the election;
and (I3) captures, inductively, the precondition of each vote forwarding action,
which is enabled by a sequence of forwards from nodes with smaller IDs.
The invariant, in effect, encodes the same causal chain that a developer would naturally describe.
But the process of \emph{arriving} at it through counterexample-driven strengthening
obscures this connection.

\textit{How developers naturally reason.}\;
Now consider how we naturally argue about this protocol's correctness.
We reason about all possible executions and the states reachable from them.
States no longer range in the whole universe, 
but only those fitting in a context will be considered:
either after an initial vote from some node, or after a sequence of forwarded votes that are causally enabled by prior ones.
Thus, we naturally, and systematically, organize the protocol's behavior into \emph{causal histories},
sequences of causally related actions, each tracing how a single vote propagates through the ring.

Consider a concrete ring of three nodes with IDs $1$, $3$, $2$ in counter-clockwise order
($1 \to 3 \to 2 \to 1$).
There are three causal histories, one originating from each node:
\begin{enumerate*}
    \item Node $1$ sends vote $1$ to node $3$.
    Since $3 > 1$, node $3$ discards it. The sequence terminates.
    \item Node $2$ sends vote $2$ to node $1$.
    Since $1 < 2$, node $1$ forwards it to node $3$.
    Since $3 > 2$, node $3$ discards it, terminating the causal history.
    \item Node $3$ sends vote $3$ to node $2$, then node $1$, finally back to node $3$.
    Node $3$ receives its own ID and is elected, marking the end of the history.
\end{enumerate*}
The correctness argument is immediate:
a causal history terminates in one of exactly two ways:
its vote is discarded when it reaches a node with a larger ID,
or it returns to its originator because no node in the ring has a larger ID.
Only the causal history from the node with the maximum ID can complete a full traversal,
and reaches a state where the initial node is elected,
ensuring that exactly one node is ever elected.

\begin{figure}[t]
\small
\begin{align*}
    \mathtt{scenario}_{\mathtt{RLE}} &\triangleq \mathtt{For}(i \in [\mathtt{min}, \mathtt{max}])\;\, \mathtt{voteChain}(i),
    \qquad
    \mathtt{voteChain}(i) \triangleq \mathtt{PhaseChain}[P_1, P_2, P_3], \\[2pt]
    P_1 &\triangleq \mathtt{Init}(i) \cdot (j \leftarrow i),\\[2pt]
    P_2 &\triangleq \bigl(\, \assm(\neg\mathtt{isMax}(\mathtt{right}(j))) \cdot \mathtt{Recv}(\mathtt{right}(j), \mathtt{id}(i)) \cdot (j \leftarrow \mathtt{right}(j)) \,\bigr)^{*},\\[2pt]
    P_3 &\triangleq \assm(\mathtt{isMax}(\mathtt{right}(j))) \cdot \mathtt{Recv}(\mathtt{right}(j), \mathtt{id}(i))
\end{align*}
\caption{The scenario $\mathtt{scenario}_{\mathtt{RLE}}$ for ring leader election
    (\cref{fig:rle-protocol}). $j$ tracks the vote-chain front;
    $\mathtt{isMax}(n)$ tests whether $n$ holds the maximum ID.}
\figlabel{rle-scenario}
\end{figure}

This reasoning is clean, intuitive, and convincing.
It proceeds by directly addressing \emph{what the protocol does},
which matches the intuition held by the protocol designers when they design the protocol.
Yet this reasoning is entirely informal.
It lacks the precision needed for formal verification:
the reasoning only considers one or a finite number of examples,
the claim that these histories cover all executions is unstated,
and the correctness argument, while intuitive, is not rigorous.
The central question that motivates our work is:
\emph{can we formalize this natural, execution-oriented reasoning
into a rigorous proof methodology?}

\textit{Formal scenario-based reasoning.}\;
We now sketch how our framework formalizes the above intuitive argument.
In our framework, we use \emph{scenario} to formally describe the executions that serve the natural reasoning.
A scenario is a structured expression resembling a regular expression over protocol actions,
extended with auxiliary variables that track execution metadata and guard control flow.
For ring leader election, the scenario is a sequence of causal histories (see \figref{rle-scenario} for the full formal scenario; the language will be defined in~\secref{scenario}).
A causal history originating from node $i$ is expressed as $\mathtt{voteChain}(i)$ in 3 phases, where:
$P_1$ executes $\mathtt{Init}(i)$ and sets the tracking variable $j := i$;
$P_2$ is a loop that forwards the vote hop by hop through nodes that are not the maximum-ID node, advancing $j$ at each step;
and $P_3$ delivers the vote to the maximum-ID node, completing the chain.
The complete scenario sequences these chains over all nodes in ID order.
This sequential $\mathtt{For}$ loop reveals another crucial feature of our approach:
the scenario does not have to express all possible interleavings, but only a set of representatives.
Here, since vote chains for distinct initial nodes act on disjoint parts of the protocol state (each only receives votes with the initial node's ID),
their actions \emph{commute}: they can be reordered without affecting feasibility or the ending state.
This means every interleaving of the chains is equivalent, up to commutativity,
so $\mathtt{scenario}_\mathtt{RLE}$ can easily select the sequential one as the representative.
We design our proof system to formally establish that this scenario is a valid proof of the safety property,
via two obligations: \emph{coverage}, every protocol execution is equivalent, up to commutativity, to one described by the scenario and, \emph{correctness}, every execution described by the scenario satisfies the safety property.
Note that coverage never mentions the safety property:
a scenario whose coverage obligation is discharged is a \emph{behavioral characterization} of the protocol itself,
established once and then reusable, unchanged, for every safety property one wishes to verify.

\myparagraph{Contributions}
In this work, we propose \emph{scenario-based proofs} for distributed protocols, a verification methodology
that formalizes exactly this execution-oriented reasoning.
The developer writes a \emph{scenario}, a structured expression describing representative
executions of the protocol, 
and uses a proof system to formally verify
a given scenario as a proof against a safety property.

\textit{Scenario Language.}\;
We first observe three key components of describing distributed protocols: the \emph{causal histories} that capture the backbones of the executions formed by causality between actions, the \emph{phases} that organize each history into sequential segments, and the \emph{sequentialization} that picks a representative ordering for causal history interleavings.
We then introduce the scenario language in \secref{scenario} to formally describe these representative executions, a regular-expression-like notation over
protocol actions and auxiliary variables whose layered denotational semantics closes the denoted
runs under a \emph{semi-commutation} relation, so a single representative ordering stands for all
equivalent interleavings.

\textit{Proof System.}\;
For a given scenario, the proof reduces to three tasks (\secref{proof-system}).
\emph{Coverage} requires that every feasible protocol execution lies in the commutativity closure of the scenario;
we discharge it by a syntax-directed calculus that rewrites on the structure of the scenario.
\emph{Correctness} requires that every execution in the commutativity closure is safe;
since commutativity-equivalent executions reach the same state,
it suffices to verify directly over the representative executions described by the scenario,
using weakest-precondition reasoning with user-supplied loop invariants.
\emph{Soundness} of the semi-commutation, which both reductions assume, is checked locally on pairs of actions.
The three proof obligations altogether ensure that the scenario is a valid proof for the protocol against the specified safety property.
Moreover, decoupling coverage and soundness from correctness brings the following benefit:
coverage and soundness are proven once per protocol and may be reused across different safety properties,
while each safety property requires only its own, typically small, correctness argument.

\textit{Case Studies.}\;
We develop scenario proofs for six distributed protocols:
\begin{enumerate*}
    \item ring leader election, the running example of \secref{scenario} and \secref{proof-system},
    \item single-decree Paxos, our main case study (\secref{casestudy}), and 
    \item two-phase commit and a decentralized lock, together with two variants of a multi-signature payment protocol from the Ivy benchmark~\cite{ivy} in the Appendix.
\end{enumerate*}

\section{Preliminaries}
\seclabel{prelim}

In this section, we first formalize distributed protocols and their semantics,
most of which is standard and familiar to the reader. 
We present them to fix notation and terminology for later precision.
Then we introduce (semi-) commutativity in trace theory~\cite{mazurkiewicz,semi-commutations}, which makes a compact representation of the protocol behaviors possible.

\subsection{Distributed Protocols}

\begin{definition}[Protocol]
\deflabel{protocol}
A \emph{protocol} $\protocol$ consists of the following components.
\begin{itemize}[nosep]
    \item A finite set $\ntypes$ of \emph{node types}.
    For each node type $t \in \ntypes$, there is a set $\lstate_t$ of \emph{local states}.
    \item A finite set $\mtypes$ of \emph{message types}, associated with three functions
    \begin{itemize}
        \item $\sender \colon \mtypes \to \ntypes$,
        \item $\receiver \colon \mtypes \to \ntypes$, and
        \item $\payload \colon \mtypes \to 2^{\domain}$,
    giving the node types of a message's sender and receiver and its payload domain.
    \end{itemize}
    \item A set $\gstate$ of \emph{shared states}.
    \item A finite set $\acts$ of \emph{actions} (\defref{action}).
    \item A set $\instances$ of valid \emph{instances} (\defref{instance}).
\end{itemize}
\end{definition}

A node type represents a class of nodes that share the same role in the protocol: the same actions and the same type of local state.
For example, ring leader election has a single node type and a single message type,
while two-phase commit has two node types (coordinator and participant) and several message types (prepare, commit, abort, etc.).
The shared state captures the shared components of the system:
in practice, $\gstate$ typically encodes the network model and
other common uses include a shared clock or shared memory.
To prove safety properties, we focus on asynchronous message passing, so $\gstate$
consists of two sets: a \emph{sent buffer} and a \emph{received buffer},
which we assume across all protocols later in the paper.
We abstract away the values that may be carried by messages, stored in states, and manipulated by actions, and use $\domain$ to denote the value domain of all such values.
Over value domain $\domain$ and node types $N\in \ntypes$, we fix a signature of function symbols $F$ (arithmetic, $\id$ and
$\id^{-1}$, set insertion, map update, field selectors $.f$ on messages, local states and the shared state, etc.) and predicate symbols $R$ ($=$, $\le$, $\in$, node topology, etc.).
These symbols are uninterpreted at the level of the protocol;
each instance supplies an interpretation for them (\defref{instance}),
constrained by the instance theory below.

\myparagraph{Actions}
Each action specifies the type of node that executes it and, optionally, the type of message that triggers it;
its semantics is a guard and an effect, both over the acting node's local state and the shared state.
A guard typically checks that the triggering message is in the sent buffer but not in the received buffer,
and that the local state satisfies certain conditions;
an effect typically records the triggering message as received, adds new messages to the sent buffer, and updates the local state.
We deliberately abstract the concrete data structure of the buffers and the membership checks,
so that the definition covers different communication models (e.g., FIFO vs.\ non-FIFO, with or without message loss).

\begin{definition}[Action]
\deflabel{action}
Each action $a \in \acts$ is associated with
a node type $\ntype(a) \in \ntypes$ and, optionally, a message type $\mtype(a) \in \mtypes \cup \{\bot\}$.
It is then defined by two functions, parameterized with a node $n$ of type $\ntype(a)$
and a message $m$ of type $\mtype(a)$ (or $m = \bot$ when $\mtype(a) = \bot$):
a \emph{guard} $\gd_a(n,\,m) \subseteq \gstate \times \lstate_{\ntype(a)}$,
determining whether the action is \emph{enabled} on a given shared and local state,
and an \emph{effect} $\eff_a(n,\,m) \colon \gstate \times \lstate_{\ntype(a)} \to
\gstate \times \lstate_{\ntype(a)}$,
returning the updated shared and local states.
\end{definition}

\myparagraph{Instances and Configurations}
A protocol is parameterized by its set of valid instances $\instances$, which constrains the allowed deployments:
ring leader election allows any finite number of nodes,
while two-phase commit requires exactly one coordinator and one or more participants.
We specify $\instances$ axiomatically: the protocol fixes an \emph{instance theory} $\thy$,
a set of first-order axioms over the signature $F \cup R$ together with the vocabulary of nodes,
messages, and the initial configuration
(e.g., the nodes form a ring; node identifiers are unique; there is exactly one coordinator),
and $\instances$ is the set of instances whose interpretation satisfies $\thy$.

\begin{definition}[Instance and Configuration]
\deflabel{instance}
An \emph{instance} of a protocol $\protocol$ is a tuple $I = (\nodes^I, \id, \tau, \interp, c_0)$,
where $\nodes^I$ is a finite set of \emph{nodes},
$\id: \nodes^I \to \nats$ maps each node to a unique identifier,
$\tau \colon \nodes^I \to \ntypes$ is a \emph{type assignment},
$\interp$ is an \emph{interpretation} that assigns to each function symbol $f \in F$
a function and to each predicate symbol $r \in R$ a relation, of the matching arity,
over $\domain \cup \nodes^I$ (and over local states, shared states and messages, for field selectors),
and $c_0$ is the initial \emph{configuration}.
A configuration is a pair $c = (g, \ell)$,
where $g \in \gstate$ is the shared state
and $\ell$ maps each node $n \in \nodes^I$ to a local state $\ell(n) \in \lstate_{\tau(n)}$.
\end{definition}

\myparagraph{Protocol Instantiation}
The semantics of a protocol is defined across all valid instances.
Fix an instance $I = (\nodes^I, \id, \tau, \interp, c_0)$.
For each node type $t\in \ntypes$, we write $\nodes_t$ for the nodes of type $t$.
The set of \emph{instantiated messages} is
\[
    \msgs_I = \bigl\{\, (m, n_1, n_2, d) \;\big|\; m \in \mtypes,\
    n_1 \in \nodes_{\sender(m)},\ n_2 \in \nodes_{\receiver(m)},\ d \in \payload(m) \,\bigr\},
\]
written $\msgs$ when $I$ is clear, with $\msgs_m$ the messages of type $m \in \mtypes$.
The set of \emph{instantiated actions} is
\[
    \actsI \;=\; \bigl\{\, (a, n, m) \;\big|\; a \in \acts,\ n \in \nodes_{\ntype(a)},\
    m \in \msgs_{\mtype(a)} \cup \{\bot\} \,\bigr\},
\]
where $(a, n, m)$ is the action $a$ executed by node $n$, optionally triggered by message $m$.

\myparagraph{Runs and Executions}
We define the semantics of a protocol in terms of its \emph{runs}.
Given an instance $I$, a \emph{run} is a sequence of instantiated actions $w \in \actsI^{*}$.
We extend runs to \emph{executions}, which are sequences of 
configuration transitions starting from the initial configuration $c_0$, and take the effect of each action in the run;
it is \emph{feasible} when every step's action is enabled.

\begin{definition}[Transition and Execution]
\deflabel{execution}
Given a configuration $c = (g, \ell)$
and an instantiated action $(a,\, n,\, m) \in \actsI$,
$c$ transitions to $c' = \bigl(g',\; \ell[n \mapsto s']\bigr)$, where
$(g', s') = \eff_a(n, m)(g, \ell(n))$,
written $c \xrightarrow{(a,\, n,\, m)} c'$.
An \emph{execution}
is a finite sequence of steps
$c_0 \xrightarrow{\alpha_1} c_1 \xrightarrow{\alpha_2} \cdots \xrightarrow{\alpha_k} c_k$
with each $\alpha_i = (a_i, n_i, m_i) \in \actsI$;
it is \emph{feasible} if $(g_{i-1}, \ell_{i-1}(n_i)) \in \gd_{a_i}(n_i, m_i)$ for all $1 \leq i \leq k$.
We write $\reach(c_0 \xrightarrow{\alpha_1} c_1 \xrightarrow{\alpha_2} \cdots \xrightarrow{\alpha_k} c_k) = c_k$ for the final configuration of an execution.
\end{definition}

Since effects are deterministic, a run induces a unique execution, 
thus we abuse notation and say a run $\pi$ is feasible if the induced execution is feasible.
We write $\reach(\pi)$ for the final configuration of the execution induced by the run $\pi$.
We use $\exec_I$ (resp.\ $\fexec_I$) to denote the set of all (resp.\ feasible) runs of $I$,
with $\reach(\epsilon) = c_0$.
The \emph{semantics} of $\protocol$ collects the feasible runs of all valid instances:
\[
    \sem{\protocol} \;=\; \bigcup_{I \in \instances} \fexec_I .
\]
Notice that we define the semantics of a protocol as a set of runs, rather than executions, because we will later introduce a (semi-)commutativity relation on runs, which is not well-defined on executions.

\myparagraph{Safety Properties}
A configuration is \emph{reachable} if it appears, as an intermediate or final configuration,
along some feasible execution of some valid instance.

\begin{definition}[Safety Property]
\deflabel{safety}
A \emph{safety property} (or \emph{invariant}) $\varphi$ is a predicate on configurations.
A protocol $\protocol$ \emph{satisfies} $\varphi$, written $\protocol \models \varphi$, if
every reachable configuration of every valid instance satisfies $\varphi$.
\end{definition}

Since every prefix of a feasible execution is itself feasible, $\sem{\protocol}$ is prefix-closed,
and every reachable configuration is the final configuration of a shorter feasible execution.
Safety therefore reduces to the final configurations.

\begin{lemma}
\lemlabel{safety-reach}
$\protocol \models \varphi$ if and only if
$\reach(\pi) \models \varphi$ for every $\pi \in \sem{\protocol}$.
\end{lemma}

\begin{proof}
Immediate from prefix closure:
every intermediate configuration of a feasible execution is the final configuration of a shorter one.
\end{proof}

\subsection{Commutativity}
\seclabel{commutativity}

Recall that the core idea of our approach is that a protocol's feasible runs can be represented compactly by a small set of \emph{canonical} runs.
The insight behind this is that many runs are related in the sense that they can be transformed into each other by reordering independent actions, and thus they reach the same final configuration.
We formalize this notion of related runs using the concept of semi-commutation.
Fix an alphabet $\Sigma$.
A \emph{semi-commutation relation} $\theta \subseteq \Sigma \times \Sigma$ is an irreflexive relation.
We write $u \rew_\theta v$ when $u = w\, a\, b\, w'$ and $v = w\, b\, a\, w'$ for some
$(a, b) \in \theta$, and $\rew_\theta^{*}$ for the reflexive--transitive closure.
The induced downward-closure operator ${\downarrow}_\theta \colon 2^{\Sigma^{*}} \to 2^{\Sigma^{*}}$
collects all words that rewrite into a given language:
\[
    {\downarrow}_\theta(L) \;=\; \bigl\{\, u \in \Sigma^{*} \;\big|\;
    u \rew_\theta^{*} v \text{ for some } v \in L \,\bigr\}.
\]
Later, we instantiate this on runs by taking the action alphabet $\actsI$ of an instance $I$
and fixing a semi-commutation relation $\theta_I \subseteq \actsI \times \actsI$.
Thus, a compact set of canonical runs $L \subseteq \actsI^{*}$
represents a larger set of runs ${\downarrow}_{\theta_I}(L)$.

\section{Scenarios}
\seclabel{scenario}

Recall from the introduction that a protocol designer argues for correctness by 
walking through the executions that can occur, and grouping them into a few representative patterns.
A \emph{scenario} is the formal counterpart of such an argument.
This section makes the notion precise, fixing a syntax in which 
a scenario is written and a semantics that captures which protocol runs it denotes.
Before the formal definitions, we motivate the design by presenting the intuition behind how one writes the representative runs of ring leader election (\figref{rle-scenario}).

\begin{example}[Scenario of Ring leader election]
\exlabel{writing-rle}
We illustrate three insights behind the design of scenarios in the context of ring leader election: causal histories, phases, and sequentialization.

\myparagraph{Causal histories} Given any feasible run of the protocol, we can usually identify several backbones from a sequence of actions, similar to threads in a concurrent program.
They are triggered and ordered in a certain way to reflect how nodes are collaborated in the protocol.
In ring leader election, these backbones are vote chains that link the actions that initiate and propagate a vote.
The order is constrained and cannot be flipped by the causal dependencies among actions, where later actions are enabled by message-passing (as in the vote chains) or state updates
\footnote{The idea descends from the happens-before order~\cite{lamport-time} and also recent message chains~\cite{message-chains}. We do not use the latter term directly but causal histories as a causal history generalizes it,
since an action may be enabled by several predecessors and the history is then a directed acyclic graph rather than a chain.}.
Therefore, instead of considering all interleavings of all actions, a protocol's behavior then decomposes into different interleavings of a few causal histories.
The first step in constructing a scenario expression becomes
collecting these dependencies and forming the run's \emph{causal history}. 
Notice that unlike concurrent programs, those causal histories, or ``threads'', are not limited to single nodes.

The later two insights transform all possible interleavings of causal histories into representative scenarios: phases capture the sequential stages within a single history, and sequentialization arranges multiple histories into a representative order.

\myparagraph{Phases} 
Since the set of protocol runs is prefix-closed, any partial execution of a protocol run is also a valid run.
Thus, shorter causal histories are usually prefixes of longer ones.
On the other hand, longer causal histories can be split into shorter segments, each representing different stages, or \emph{phases} in the execution.
The notion of phases is a common concept in the design of protocols, such as two-phase commit and the prepare/propose phases in Paxos.
Instead of writing causal histories of different lengths separately, 
we identify these phases and systematically group causal histories that correspond to different phases of the same execution.
The following is how we describe a causal history with several phases in ring leader election.
A vote chain has three phases: initiation ($P_1$, the $\mathtt{Init}$), relay past non-maximum nodes ($P_2$, an iteration of $\mathtt{Recv}$s), and final delivery to the maximum node ($P_3$).

\myparagraph{Sequentialization}
The whole protocol contains several causal histories, and the histories run \emph{concurrently}: their actions interleave in combinatorially many ways, which makes reachable states hard to reason about.
A scenario does not enumerate them but \emph{sequentializes} them into several representative orders that rely on a semi-commutativity relation.
This might be the hardest and creative part of designing scenarios.
In ring leader election, we notice that each vote chain is completely independent of the others, i.e., actions in one vote chain do not affect the state or enabledness of actions in another vote chain.
Thus, we choose the naive way of sequentialization here, arranging the whole vote chains one after another in an arbitrary but fixed order (e.g., by their originators' identifiers).
In~\secref{casestudy}, we will see that this naive sequentialization is not always sufficient for more complex protocols (e.g., Paxos), but a clever one is required.
\end{example}

The example above does not only guide our intuition on writing a scenario, but also motivates the syntax definition (\secref{scenario-syntax}), where we need nondeterminism to capture different stop points on phases, loops to represent repeated patterns (such as the vote relay, or vote chains on different originators), and also auxiliary variables to control the loops.
We then define the semantics (\secref{scenario-semantics}), and finally return in \secref{writing-scenarios} to the derived syntax sugars that make such scenarios convenient to write.

\subsection{Scenario Syntax}
\seclabel{scenario-syntax}

A scenario denotes runs of an \emph{arbitrary} valid instance $I \models \thy$.
We therefore fix such an instance $I$ and build everything on it.

\myparagraph{Scenario variables and values}
A scenario works as a global orchestrator of the run, 
tracking the meta-information to guide which action is fired next.
To track the such information a run depends on, 
a scenario carries its own state over a set of auxiliary \emph{scenario variables}, which we fix here as $\scvars$.
Besides its own state, a scenario may also \emph{read}, but never write, the current protocol configuration;
only protocol actions change the configuration.

\myparagraph{Expressions}
We first define the values and boolean expressions over the scenario variables and the protocol configuration:
\[
    e \;::=\; v\in \scvars \;\mid\; d \in \dom_I \;\mid\; F(e_1, \dots, e_k) \;\mid\; \mathsf{glob} \;\mid\; \mathsf{loc}(e),
    \qquad
    b ::= \mathsf{true} \;\mid\; R(e_1, \dots, e_k) \;\mid\; \neg b \;\mid\; b \land b ,
\]
with $F, R$ of matching arity.
Constants $d$ range over the values of every type: elements of $\dom_I$, nodes, and messages.
The expression $\mathsf{glob}$ reads the current shared state and $\mathsf{loc}(e)$ reads the current local state of the node $e$.
Their components are accessed by field selectors $.f$, which are function symbols of $F$ like any other
(e.g., $\mathsf{glob}.\mathtt{sent}$, $m.\mathtt{id}$);
we abbreviate $\mathsf{loc}(e).f$ as $e.f$ when $e$ denotes a node.

\myparagraph{Command terms}
Scenarios are built over a single alphabet $\Sigma_I$ of \emph{command terms} of three kinds:
\begin{itemize}[nosep]
    \item a \emph{protocol action} $a(e_n, e_m)$, where $a \in \acts$ and the expressions
    $e_n$ and $e_m$ select the executing node and the triggering message
    (we omit $e_m$ and write $a(e_n)$ when $\mtype(a) = \bot$),
    \item an \emph{assignment} $v \leftarrow e$, where $v \in \scvars$ and $e$ is an expression,
    \item a \emph{test} $\assm(b)$, where $b$ is a Boolean expression.
\end{itemize}
We assume command terms are \emph{type-correct}: $e_n$ evaluates to a node of type $\ntype(a)$ and
$e_m$ to a message of type $\mtype(a)$. 
Moreover, every variable $v \in \scvars$ has a fixed type (a node type, a message type, or a sort of $\dom$), whose set of values in $I$ we write $\dom_v$, ensuring that assignments to it are type-correct.

\myparagraph{Scenario expressions}
A \emph{scenario} is a regular expression over the command terms with the two constants $\mathbf{1}$ and $\mathbf{0}$,
extended with a \emph{havoc} construct:
\[
    s \;::=\; \mathbf{1}\;\mid\;\mathbf{0}\;\mid\; t \;\mid\; \mathtt{havoc}(v) \;\mid\; s + s \;\mid\; s \cdot s \;\mid\; s^{*},
    \qquad t \in \Sigma_I,\ v \in \scvars,
\]
where $+$ is nondeterministic choice, $\cdot$ sequential composition,
$(-)^{*}$ iteration, $\mathbf{1}$ the \emph{unit} denoting only the
empty run, and $\mathbf{0}$ the \emph{zero} denoting no run at all.
$\mathtt{havoc}(v)$ assigns $v$ an arbitrary value of its type:
it is the nondeterministic choice among the assignments $v \leftarrow d$ for all $d \in \dom_v$.
Since $\dom_v$ may be infinite, this choice cannot in general be written with the finite $+$.
We use $\mathtt{Scenario}_{\Sigma_I}$ to denote the set of all scenario expressions over the alphabet $\Sigma_I$.

\subsection{Scenario Semantics}
\seclabel{scenario-semantics}
A scenario expression $s \in \mathtt{Scenario}_{\Sigma_I}$ is defined over the \emph{scenario alphabet} $\Sigma_I$ of command terms and only the representative runs, but its semantics $\sem{s}_{I, \sigma_0}$, over the \emph{protocol alphabet} $\actsI$ of concrete runs, is meant to capture the set of all possible runs.
We bridge the gap with a series of semantic layers that progressively interpret scenario expressions over $\Sigma$ into sets of runs over $\actsI$.

We first define a scenario state $\sigma : \scvars \to \dom$ as a mapping from scenario variables to their values in the domain.
Expressions are evaluated in a \emph{joint state} $(\sigma, c)$ of a scenario state and a configuration $c = (g, \ell)$:
$\evalx{v}{\sigma, c} = \sigma(v)$, $\evalx{d}{\sigma, c} = d$,
$\evalx{F(e_1, \dots, e_k)}{\sigma, c} = \interp(F)(\evalx{e_1}{\sigma, c}, \dots, \evalx{e_k}{\sigma, c})$,
$\evalx{\mathsf{glob}}{\sigma, c} = g$, and $\evalx{\mathsf{loc}(e)}{\sigma, c} = \ell(\evalx{e}{\sigma, c})$;
we write $(\sigma, c) \models b$ for the truth of a Boolean expression.
The semantics of scenario expressions is parameterized by an initial scenario state $\sigma_0$ and the instance $I$, which also fixes the initial configuration $c_0$.
When there is no ambiguity, we will omit these parameters and simply write $\sem{s}$.

\begin{itemize}[nosep]
    \item The \textbf{structure} layer $\semSig{\cdot}: \mathtt{Scenario}_{\Sigma_I} \to 2^{\Sigma_I^*}$ unfolds a scenario based on the standard regular-expression semantics to produce a language of command-term sequences over $\Sigma_I$.
    \item The \textbf{variable} layer $\semvar{\cdot}: 2^{\Sigma_I^*} \to 2^{\actsI^*}$ binds the scenario state and the protocol configuration, resolving each command term into a concrete action and so turning command-term sequences into runs over $\actsI$.
    \item The \textbf{feasibility} layer $\semfeas{\cdot}: 2^{\actsI^*} \to 2^{\actsI^*}$ intersects with the feasible runs
    $\fexec_I$, keeping only the resolved runs the protocol can actually take.
    \item The \textbf{commutativity} layer $\semcom{\cdot}: 2^{\actsI^*} \to 2^{\actsI^*}$ takes the semi-commutative closure under $\theta_I$ so that the scenario stands for every run that can be reordered into one of its own.
\end{itemize}
Thus, the semantics of a scenario is a composition of four layers, each an operator on the output of the one below:
\[
    \sem{s} \;=\; \semcom{\semfeas{\semvar{\semSig{s}}}}.
\]



\myparagraph{Structure interpretation}
For a given scenario $s$, we define
\[
    \semSig{s} =
    \begin{cases}
        \{\epsilon\} & s \equiv \mathbf{1},\\
        \emptyset & s \equiv \mathbf{0},\\
        \{t\} & s \equiv t \in \Sigma,\\
        \{\, (v \leftarrow d) \mid d \in \dom_v \,\} & s \equiv \mathtt{havoc}(v),\\
        \semSig{s_1} \cup \semSig{s_2} & s \equiv s_1 + s_2,\\
        \semSig{s_1} \cdot \semSig{s_2} & s \equiv s_1 \cdot s_2,\\
        \bigcup_{n \ge 0} \bigl(\semSig{s_1}\bigr)^{n} & s \equiv s_1^{*},
    \end{cases}
\]
where $L_1 \cdot L_2 = \{\, w_1 w_2 \mid w_1 \in L_1,\ w_2 \in L_2 \,\}$ is language
concatenation and the powers are $L^{0} = \{\epsilon\}$ and $L^{n+1} = L^{n} \cdot L$.

\myparagraph{Variable interpretation}
Similar to the standard denotational semantics for imperative programs, the variable interpretation transforms a command term sequence into a set of concrete protocol action sequences (runs) by evaluating expressions in the current joint state and updating it accordingly:
assignments update the scenario state, and protocol actions update the configuration.

We first define the resolution of a single word $w \in \Sigma^{*}$ from a joint state $(\sigma, c)$ inductively:
\[
    \runl(\epsilon, \sigma, c) = \{\epsilon\},
    \qquad
    \runl(t \cdot w, \sigma, c) =
    \begin{cases}
        \alpha \cdot \runl(w, \sigma, c')
            & t \equiv a(e_n, e_m),\ \alpha = \bigl(a,\ \evalx{e_n}{\sigma, c},\ \evalx{e_m}{\sigma, c}\bigr),\ c \xrightarrow{\alpha} c',\\
        \runl\bigl(w,\ \sigma[v \mapsto \evalx{e}{\sigma, c}],\ c\bigr) & t \equiv (v \leftarrow e),\\
        \runl(w, \sigma, c) & t \equiv \assm(b),\ (\sigma, c) \models b,\\
        \emptyset & t \equiv \assm(b),\ (\sigma, c) \not\models b,
    \end{cases}
\]
We then lift it to a language $L \subseteq \Sigma^{*}$:
$\semvar{L} = \bigcup_{w \in L} \runl(w, \sigma_0, c_0)$.
The transition $c \xrightarrow{\alpha} c'$ applies the effect of $\alpha$ without consulting its guard (\defref{execution}),
so resolution is independent of feasibility.
Because effects are deterministic, the configuration a term reads is exactly $\reach(c_0, \pi)$,
where $\pi$ is the run emitted before it.
Thus, for every run that survives the feasibility cut below, each term has read the configuration of that run's own feasible execution.



\myparagraph{Feasibility cut}
The output of variable interpretations gives a set of concrete protocol runs, 
which might include infeasible runs that the protocol cannot actually execute.
Feasibility cut intersects the output with the
feasible runs $\fexec_I$ of $I$.
For a language $L \subseteq \actsI^{*}$, the feasibility layer is
$\semfeas{L} = L \cap \fexec_I$.

\myparagraph{Commutativity closure}
The commutativity layer closes the feasible runs under the semi-commutation relation
$\theta_I$ (\secref{commutativity}).
For a language $L \subseteq \actsI^{*}$,
$\semcom{L} = {\downarrow}_{\theta_I}(L)\cap \fexec_I$.

In summary, every scenario $s$ now denotes a set $\sem{s} \subseteq \actsI^{*}$ of (feasible) runs.
Variable interpretation is applied first to obtain the set of concrete runs from the scenario, and then the commutativity closure is applied to augment the representation with all runs that are semi-commutatively equivalent.

\subsection{Syntax Sugars}
\seclabel{writing-scenarios}

We close this section by formally introducing the two syntax sugars that make it easier to write scenarios: 
$\mathtt{PhaseChain}$, which lays out one history as a sequence of \emph{phases}, 
and $\mathtt{For}$, which composes histories over a range of \emph{nodes}.
These complete the formality of ring leader election scenario (\figref{rle-scenario}).

\myparagraph{Phases and \texorpdfstring{$\mathtt{PhaseChain}$}{PhaseChain}}
Given a causal history that can be divided into phases: $s_1, \dots, s_k$, we simply write
\[
    \mathtt{PhaseChain}[s_1, \dots, s_k]
    \;\triangleq\;
    \mathbf{1} + s_1 + s_1\cdot s_2 + \cdots + s_1\cdot s_2 \cdots s_k
\]
to denote the union of every prefix, including the empty run $\mathbf{1}$.

\myparagraph{Nodes and \texorpdfstring{$\mathtt{For}$}{For}}
\seclabel{node-model}
Composing one causal history per node supposes a way to range over nodes in order.
We fix a node type $N\in \ntypes$ and let the instance theory $\thy$ constrain its node set $\mathtt{N}_N$ to be finite and linearly ordered, from $\mathtt{min}$ to $\mathtt{max}$ with successor $\mathtt{next}$ (inverse $\mathtt{prev}$).
Finiteness is not a restriction in practice, as any real network has finitely many nodes.
To iterate over all nodes in $\mathtt{N}_N$,
inspired by the standard for-loop construct in programming, 
we sequence one copy of $s$ per node over the core grammar:
\[
    \mathtt{For}\bigl(v \in [\mathtt{min}, \mathtt{max}]\bigr)\ s(v)
    \;\triangleq\;
    (v \leftarrow \mathtt{min})\cdot
    \bigl(\assm(v \le \mathtt{max})\cdot s(v)\cdot(v \leftarrow \mathtt{next}(v))\bigr)^{*}\cdot
    \assm(v > \mathtt{max}),
\]
initializing $v$ to $\mathtt{min}$, iterating the body and advancing $v$ while $v \le \mathtt{max}$,
and checking on exit that the loop ran to completion (so $v$ has passed $\mathtt{max}$).
Finally, notice that this order is a modeling device, not the communication topology (e.g., the ring structure in ring leader election).
We always assume that the body $s$ never assigns to the iterator $v$, which can be easily checked syntactically and ensures the correctness of the unrolling rule in~\secref{coverage}

\section{The Scenario Proof System}
\seclabel{proof-system}

In this section, we introduce the \emph{syntex-directed} scenario proof system, which allows us to reason about safety properties of protocols using a scenario expression and three sub-tasks: coverage, correctness, and soundness of the semi-commutation.

\subsection{Overview}
The starting point of the proof is a simple observation about how one proves a safety property.
To show that every protocol run of an instance $I$ reaches a configuration that satisfies $\varphi$, it is
enough to exhibit a set of runs that, on the one hand, accounts for every feasible run and,
on the other, ends only in $\varphi$-states.
We call such a set of runs a \emph{witness} for the safety property $\varphi$.

\begin{theorem}[Witness]
\thmlabel{witness}
Let $\varphi$ be a safety property.
For any valid $I\models \thy$,
if there exists a set of runs $S_I \subseteq \actsI^{*}$ that satisfies
\begin{itemize}[nosep]
    \item \textbf{(Coverage)} $\fexec_I \subseteq S_I$, and
    \item \textbf{(Correctness)} $\reach(c_0, w) \models \varphi$ for every feasible $w \in S_I$,
\end{itemize}
then $\protocol \models \varphi$.
\end{theorem}

\begin{proof}
By \lemref{safety-reach} it suffices to check $\reach(c_0, w) \models \varphi$ for every
feasible run $w$.
Every such $w$ lies in $\fexec_I \subseteq S_I$, so correctness applies.
Quantifying over all instances gives $\protocol \models \varphi$.
\end{proof}

A scenario provides such a witness: $S_I = \sem{s}_{I, \sigma_0}$.
Now the proof obligations can be stated in two parts:
\begin{enumerate*}
    \item coverage: $\fexec_I \subseteq \sem{s}$,
    \item correctness: $\reach(c_0, w) \models \varphi$ for every $w \in \sem{s}$.
\end{enumerate*}
But now the two obligations speak about the full closure $\sem{s}$ rather than about
the compact $s$,
violating the principle that a scenario proof can rely on the simple notation of $s$ rather than the expanded closure $\sem{s}$.
We resolve this by a key observation:
when the semi-commutation $\theta_I$ is \emph{sound} (defined later), 
both obligations can be discharged on $s$ using \emph{syntex-directed} approaches, instead of the semantic closure $\sem{s}$.
This yields a third proof task:
the soundness of $\theta_I$.


We first define the soundness of $\theta_I$ and then introduce how we obtain syntex-directed proof system on both coverage and correctness obligations.

\begin{definition}[Sound semi-commutation]
\deflabel{sound-theta}
The semi-commutation $\theta_I$ is \emph{sound} if every $(a, b) \in {\theta_I}$,
where we have $ab \rew_{\theta_I} ba$, satisfies that, for any configuration $c$, if $ab$ is feasible on $c$, then
\begin{itemize}[nosep]
    \item \textbf{feasibility-preserving}: $ba$ is also feasible on $c$, and
    \item \textbf{reach-preserving}: $\reach(c, ab) = \reach(c, ba)$.
\end{itemize}
\end{definition}


\subsubsection{Reducing coverage checking}
We first observe that checking coverage $\fexec_I \subseteq \sem{s}$ is hard as the left-hand side $\fexec_I$ is hard to characterize explicitly.
Our first step is to transform this obligation into a form that is independent of $\fexec_I$.
Observe that $\fexec_I$ is the smallest extension-closed subset of $\actsI^*$ (defined below). 
Therefore, in order to show coverage, it suffices to show that $\sem{s}$ contains the smallest extension-closed set, i.e., that $\sem{s}$ is also extension-closed.

\begin{definition}[Extension-closed]
\deflabel{extension-closed}
A set of runs $L \subseteq 2^{\actsI^*}$ is \emph{extension-closed} if $\epsilon \in L$
and, for every $w \in L$ and every action $a \in \actsI$ with $w \cdot a \in \fexec_I$, also
$w \cdot a \in L$.
\end{definition}

\begin{lemma}[Coverage via closedness]
\lemlabel{coverage-closed}
If $\sem{s}$ is extension-closed, then coverage holds.
\end{lemma}

\begin{proof}
    The proof directly follows from the fact that $\fexec_I$ is the smallest extension-closed set, and ``smallest'' is defined over set inclusion.
\end{proof}

However, the step condition of extension-closedness, $\sem{s} \cdot a \cap \fexec_I\subseteq \sem{s}$, is not syntax-directed: the left-hand side is assembled by appending $a$ to a denotation, so a syntax-directed system has no scenario structure to recurse on.
This is where the soundness of $\theta_I$ comes in: it allows us to replace the left-hand side by a denotation of a scenario, yielding a syntax-directed obligation.

\begin{lemma}[Closedness reduction]
\lemlabel{closedness-reduction}
If $\theta_I$ is sound, then $\sem{s}$ is extension-closed if and only if
\[
    \sem{\mathbf{1}} \subseteq \sem{s}
    \qquad\text{and}\qquad
    \sem{s \cdot a} \subseteq \sem{s} \ \text{ for every action } a \in \actsI .
\]
\end{lemma}

\begin{proof}[Proof sketch]
The base $\epsilon \in \sem{s}$ is exactly $\sem{\mathbf{1}} \subseteq \sem{s}$.
Then we prove the step obligation.
From absorption to extension-closedness. Suppose $w \in \sem{s}$ and $w\cdot a$ is feasible. 
There is a feasible representative $v$ of $s$ with $w \rew_\theta^{*} v$. 
Soundness preserves feasibility and final state along this rewrite, 
so $v\cdot a$ is feasible as well. 
Thus $w\cdot a \rew_\theta^{*} v\cdot a$ places $w\cdot a$ in $\sem{s\cdot a} \subset \sem{s}$.
From extension-closedness to absorption. 
Suppose $u \in \sem{s\cdot a}$. Then $u$ rewrites to a feasible $v\cdot a$,
where $v$ is a representative of $s$. 
Since $v \in \sem{s}$, extension-closedness gives $v\cdot a \in \sem{s}$. 
Downward closure then gives $u \in \sem{s}$.
\end{proof}

The value of this reduction is the shape of the resulting obligation, $\sem{s \cdot a} \subseteq \sem{s}$, where both sides are denotations of scenarios
and the proof obligation matches an inclusion of the uniform form $\sem{s_1} \subseteq \sem{s_2}$.
We design a syntax-directed proof system to discharge such inclusions
(\secref{coverage}).

\begin{remark}[Coverage is property-independent]
\remlabel{coverage-independent}
A scenario together with its coverage proof is of standalone interest: it makes no mention of
the safety property, and thus characterizes the behavior of the protocol itself, where every
feasible run is, up to commutativity, one the scenario describes.
Establishing a particular safety property is then one \emph{use} of this characterization,
discharged by the correctness check alone.
The same scenario and its coverage proof, established once, could serve to verify other safety
properties, or to reason about the protocol beyond safety; we leave these directions to future
work.
\end{remark}

\subsubsection{Reducing correctness checking}

In the same spirit of partial order reduction and many other techniques~\cite{dpor,godefroid-por,optimal-dpor} that exploit (sound) commutativity, correctness checking quantifies over all runs in $\sem{s}$, but it suffices to check only the representatives $\semfeas{s}$ that $s$ generates directly without reordering.

\begin{lemma}[Correctness reduction]
\lemlabel{correctness-reduction}
If $\theta_I$ is sound, then correctness holds if and only if
$\reach(c_0, w) \models \varphi$ for every feasible representative $w \in \semfeas{s}$.
\end{lemma}

\begin{proof}[Proof sketch]
Every $v \in \sem{s}$ rewrites to some representative $w \in \semfeas{s}$.
By feasibility- and reach-preservation, $w$ is feasible and $\reach(c_0, v) = \reach(c_0, w)$,
so $\varphi$ holds at $v$ exactly when it holds at $w$.
\end{proof}

In summary, 
given a scenario $s$ and a semi-commutation $\theta_I$,
our scenario-based proof is reduced to the following three tasks:
\begin{itemize}[nosep]
    \item \emph{Correctness} (\secref{correctness}): the assertion
    $\reach(c_0, w) \models \varphi$ on the representatives $\semfeas{s}$.
    \item \emph{Coverage} (\secref{coverage}): the inclusions $\sem{\mathbf{1}} \subseteq \sem{s}$
    and $\sem{s \cdot a} \subseteq \sem{s}$ for every action $a\in\actsI$.
    \item \emph{Soundness} of $\theta_I$, which the above two tasks assume.
\end{itemize}

\begin{example}[Ring leader election: semi-commutation and its soundness]
\exlabel{rle-theta}
Throughout this section, we examine how ring leader election (\exref{writing-rle}) is proven using scenario and our proof system. 
We begin with the semi-commutation relation and its soundness.
The semi-commutation $\theta$ declares two actions independent exactly when they
concern \emph{different} votes, as in the table below (read row-then-column; each cell gives the
condition under which the ordered pair is in $\theta$):
\[
\begin{array}{r|cc}
    & \mathtt{Init}(i_2) & \mathtt{Recv}(n_2, m_2) \\
    \hline
    \mathtt{Init}(i_1) & i_1 \ne i_2 & \mathtt{id}(i_1) \ne m_2.\mathtt{id} \\
    \mathtt{Recv}(n_1, m_1) & m_1.\mathtt{id} \ne \mathtt{id}(i_2) & m_1.\mathtt{id} \ne m_2.\mathtt{id}
\end{array}
\]
Actions on the \emph{same} vote, which are causally ordered in one chain, stay
dependent.
Soundness is ensured by the following facts: two actions on
different votes read and write disjoint state and are never involved in message passing.
\end{example}

\subsection{Correctness: Hoare Logic}
\seclabel{correctness}

By \lemref{correctness-reduction} correctness reduces to checking $\reach(c_0, w) \models \varphi$
on the representatives $w \in \semfeas{s}$.
These follow the syntax of $s$ and the postcondition $\varphi$ can be checked by standard Hoare-logic-style of proof~\cite{hoare-logic}.

Even though the final safety property $\varphi$ is a predicate defined on the protocol configuration $c$ (shared and local states), 
the execution of a scenario $s$ also depends on the scenario state $\sigma$ that tracks the scenario variables.
So here, we use the standard weakest precondition $\mathrm{wp}_s(Q)$ of a postcondition $Q$ on the joint state $(\sigma, c)$.
For an action term $t = a(e_n, e_m)$, we write $t_{\sigma, c} = (a, \evalx{e_n}{\sigma, c}, \evalx{e_m}{\sigma, c})$ for its resolution
and $\mathrm{en}_t(\sigma, c)$ for the enabledness of $t_{\sigma, c}$ at $c$.
\[
\begin{array}{rcl}
    \mathrm{wp}_{t}(Q) &=& \lambda (\sigma, c).\; \mathrm{en}_t(\sigma, c) \Rightarrow Q(\sigma, c') \quad \text{where } t = a(e_n, e_m),\ c \xrightarrow{t_{\sigma, c}} c',\\[0.5ex]
    \mathrm{wp}_{v \leftarrow e}(Q) &=& \lambda (\sigma, c).\;Q\bigl(\sigma[v \mapsto \evalx{e}{\sigma, c}],\, c\bigr),\\[0.5ex]
    \mathrm{wp}_{\assm(b)}(Q) &=& \lambda (\sigma, c).\;(\sigma, c) \models b \Rightarrow Q(\sigma, c),\\[0.5ex]
    \mathrm{wp}_{\mathtt{havoc}(v)}(Q) &=& \lambda (\sigma, c).\;\forall d \in \dom_v.\; Q\bigl(\sigma[v \mapsto d],\, c\bigr),\\[0.5ex]
    \mathrm{wp}_{s_1 \cdot s_2}(Q) &=& \mathrm{wp}_{s_1}\bigl(\mathrm{wp}_{s_2}(Q)\bigr),\\[0.5ex]
    \mathrm{wp}_{s_1 + s_2}(Q) &=& \mathrm{wp}_{s_1}(Q) \wedge \mathrm{wp}_{s_2}(Q)\\[0.5ex]
    \mathrm{wp}_{s_1^{*}}(Q) &=& J \quad \text{if }J \Rightarrow \mathrm{wp}_{s_1}(J) \text{ and } J \Rightarrow Q
\end{array}
\]
\begin{remark}
    For iteration, we require user-supplied loop invariants.
    Compared with the original inductive invariants used to show the safety property, that must hold under all possible action interleavings, such loop invariants are typically easier to find, as the scenario is treated as a sequential program.
    We will see this in the case studies of \secref{casestudy}. 
\end{remark}

\begin{lemma}[Correctness via Hoare logic]
\lemlabel{correctness-hoare}
If $\mathrm{wp}_s(\lambda (\sigma, c).\; \varphi(c))(\sigma_0, c_0)$ is true, then $\reach(c_0, w) \models \varphi$ for every feasible representative $w \in \semfeas{s}$, and hence, by \lemref{correctness-reduction}, correctness holds.
\end{lemma}

\begin{example}[Ring leader election: correctness]
\exlabel{rle-correctness}
The postcondition is $\varphi \equiv$ ``at most one node is elected''.
The outer $\mathtt{For}$ of $\mathtt{scenario}_{\mathtt{RLE}}$ carries the invariant
\[
    J_{\mathtt{out}} \;\equiv\;
    \forall n.\ n.\mathtt{isElected} \Rightarrow \mathtt{isMax}(n),
\]
that only the maximum node is ever elected, which entails $\varphi$ since two maximum nodes are equal.
Preserving it across one chain reduces to a loop invariant on the relay phase $P_2$, stated over its scenario variable $j$ (the front of the vote) and asserting that the vote never reaches back to $i$, but staying between $i$ and the maximum node, thus no node is elected in $P_2$:
\[
    J_{\mathtt{in}} \;\equiv\;
    \mathtt{Between}(i,\, \mathtt{right}(j),\, \mathtt{max}) \,\vee\, \mathtt{isMax}(i).
\]
At the loop exit the $P_3$ guard $\mathtt{isMax}(\mathtt{right}(j))$ combines with $J_{\mathtt{in}}$
to force the electing receive to occur only at the maximum node, re-establishing $J_{\mathtt{out}}$.
\end{example}

\subsection{Coverage: A Sub-proof System for Inclusion}
\seclabel{coverage}

\lemref{closedness-reduction} reduces coverage obligation to proving the semantic inclusions
$\sem{s\cdot a} \subseteq \sem{s}$,
where the scenario $s$ can \emph{absorb} any appended action $a$.
In this subsection, we give a more general syntax-directed proof system for deriving inclusions between two scenarios, $\sem{s_1} \subseteq \sem{s_2}$.

The remaining part of the section introduces four sets of the proof rules, organized by the layers of the semantics: structure $\incSig$, variable $\incvar$, feasibility $\incfeas$, and commutativity $\inccom$,
and show their soundness: $s_1 \inclayer s_2$ implies $\sem{s_1} \subseteq \sem{s_2}$.
Organizing the proof rules by layers helps understand how the rules are developed and verify their soundness at the underlying semantic layers modularly. 
For example, since the structure layer corresponds to the pure Kleene-algebra identities, the proof rules with Kleene Algebra reasoning are natural and sound at this layer. 
Moreover, the lifting lemma (\lemref{lifting}) ensures that the soundness of the proof rules at the lower layers (e.g., $s_1 \incSig s_2$ implies $\semSig{s_1} \subseteq \semSig{s_2}$) can be lifted to the top layer, guaranteeing $\sem{s_1} \subseteq \sem{s_2}$.

\begin{lemma}[Lifting]
\lemlabel{lifting}
An inclusion proved at a lower layer holds at every higher layer:
\[
    \semSig{s_1} \subseteq \semSig{s_2} \;\Longrightarrow\; \semvar{s_1} \subseteq \semvar{s_2} \;\Longrightarrow\;
    \semfeas{s_1} \subseteq \semfeas{s_2} \;\Longrightarrow\; \sem{s_1} \subseteq \sem{s_2} .
\]
\end{lemma}


\myparagraph{Structure rules (Kleene algebra)}
The first group reasons purely about $\semSig{\cdot}$, the regular language a scenario
denotes, blind to variables, commutativity, and feasibility.
These are the Kleene-algebra laws~\cite{kozen-kat,kozen-completeness}, proving $\incSig$ and the induced equivalences (we write
$s_1 \equiv_\Sigma s_2$ for mutual inclusion):
\[
\begin{array}{ll}
    \textsc{(Idem)}        & s + s \equiv_\Sigma s \\
    \textsc{(Comm)}        & s_1 + s_2 \equiv_\Sigma s_2 + s_1 \\
    \textsc{(Assoc-}{+}\textsc{)} & (s_1 + s_2) + s_3 \equiv_\Sigma s_1 + (s_2 + s_3) \\
    \textsc{(Assoc-}{\cdot}\textsc{)} & (s_1 \cdot s_2) \cdot s_3 \equiv_\Sigma s_1 \cdot (s_2 \cdot s_3) \\
    \textsc{(Unit)}        & \mathbf{1} \cdot s \equiv_\Sigma s \equiv_\Sigma s \cdot \mathbf{1} \\
    \textsc{(Annih)}       & \mathbf{0} \cdot s \equiv_\Sigma \mathbf{0} \equiv_\Sigma s \cdot \mathbf{0} \\
    \textsc{(Distrib)}     & s \cdot (t_1 + t_2) \equiv_\Sigma s \cdot t_1 + s \cdot t_2,
                             \quad (t_1 + t_2) \cdot s \equiv_\Sigma t_1 \cdot s + t_2 \cdot s\\
    \textsc{(Embed)}       & s_i \incSig s_1 + \cdots + s_n \\
    \textsc{(Havoc)}       & (v \leftarrow d) \incSig \mathtt{havoc}(v), \quad d \in \dom_v \\
    \textsc{(Unfold)}      & s^{*} \equiv_\Sigma \mathbf{1} + s \cdot s^{*} \equiv_\Sigma \mathbf{1} + s^{*} \cdot s \\
    \textsc{(Bot)}         & \mathbf{0} \incSig s\\
    \textsc{(Fix1)}     & s \cdot t \incSig t \;\Longrightarrow\; s^{*} \cdot t \incSig t \\
    \textsc{(Fix2)}    & t \cdot s \incSig t \;\Longrightarrow\; t \cdot s^{*} \incSig t
\end{array}
\]

The structure layer enjoys \emph{precongruence}: if $s_1 \incSig s_2$, then
$C[s_1] \incSig C[s_2]$ for every context $C$ built from $\cdot$, $+$, and $(-)^{*}$, because $\semSig{\cdot}$ is compositional and concatenation, union, and iteration are monotone.
We apply this implicitly in the calculus, rewriting subterms in place.
However, the higher layers usually do not enjoy this property, as their rules are usually context-sensitive, i.e., they can only be applied when the entry state satisfies a certain condition.
We start by a typical context-sensitive rule, the feasibility rule \textsc{Disabled}, where a scenario is equivalent to $\mathbf{0}$ if an action in the scenario is not enabled after the prefix.

\begin{lemma}[Soundness of Structure rules]
    \lemlabel{soundness-structure}
    For two scenarios $s_1, s_2$, if $s_1 \incSig s_2$ is derivable by the structure rules, then $\semSig{s_1} \subseteq \semSig{s_2}$.
    By \lemref{lifting}, we have $\sem{s_1} \subseteq \sem{s_2}$.
\end{lemma}

\myparagraph{Unfolding rule for $\mathtt{For}$}
A variable-level rule concerns the derived form $\mathtt{For}$
(\secref{node-model}).
We split the iteration at a chosen node $k$ into three pieces: the prefix up to $k$, the single copy at $k$, and the suffix after $k$.
This is useful when the appended action $a$ after a loop needs to be absorbed by a certain iteration,
which we will see in the examples later.
\begin{align*}
    \textsc{(Unfold-For)}\quad
    \mathtt{For}\bigl(v \in [\mathtt{min}, \mathtt{max}]\bigr)\, s(v)
    \;\incvar\;
    &\mathtt{For}\bigl(v \in [\mathtt{min}, \mathtt{prev}(k)]\bigr)\, s(v)
    \,\cdot\, \\
    &(v \leftarrow k)\cdot s(v) \,\cdot\,\\
    &\mathtt{For}\bigl(v \in [\mathtt{next}(k), \mathtt{max}]\bigr)\, s(v)
\end{align*}

\begin{lemma}[Soundness of $\mathtt{For}$ unfolding rule]
    \lemlabel{soundness-unfold-for}
    For two scenarios $s_1, s_2$, if $s_1 \incvar s_2$ is derivable by the $\mathtt{For}$ unfolding rule, then $\semvar{s_1} \subseteq \semvar{s_2}$.
    By \lemref{lifting}, we have $\sem{s_1} \subseteq \sem{s_2}$.
\end{lemma}

\myparagraph{Feasibility rule (\textsc{Disabled})}
The feasibility layer prunes a scenario whose appended action can never fire.
By the $\mathrm{wp}$ calculation of \secref{correctness}, ``$a$ is dead after prefix
$\mathit{pre}$'' is a checking that the postcondition $\neg\,\mathrm{en}_a$ holds at the end
of $\mathit{pre}$, evaluated from the initial state $(\sigma_0, c_0)$.
When it holds, the whole run through $a$ is infeasible, so the scenario is $\mathbf{0}$:
\[
    \textsc{Disabled}\quad
    \frac{
          (\sigma_0, c_0) \models \mathrm{wp}_{\mathit{pre}}(\neg\, \mathrm{en}_a)}
         {\mathit{pre} \cdot a \cdot \mathit{post} \;\incfeas\; \mathbf{0}}
\]
The premise is context-sensitive: it speaks about the whole prefix $\mathit{pre}$ from the initial
state, not about $a$ alone.

\begin{lemma}[Soundness of Feasibility rule]
    \lemlabel{soundness-disabled}
    For two scenarios $s_1, s_2$, if $s_1 \incfeas s_2$ is derivable by the \textsc{Disabled} rules, then $\semfeas{s_1} \subseteq \semfeas{s_2}$.
    By \lemref{lifting}, we have $\sem{s_1} \subseteq \sem{s_2}$.
\end{lemma}

\myparagraph{Term-manipulation rules}
We keep our focus on the feasibility layer. 
Three rules manipulate the scenario terms on $s_1$ to match the shape of $s_2$.
Remember that the assignment and test terms are \emph{ghost} updates and are resolved in the variable layer, so they are no longer visible at a higher layer,
which provides the freedom to insert them as long as their writes are never observed by the rest of the scenario.

\textsc{Assign-Intro} inserts a \emph{benign} ghost assignment $v \leftarrow f$ at a chosen
point (e.g.\ absorbing an appended $\mathtt{Init}$ into $P_1$ by adding $j \leftarrow i$).
Benign means the written value is never \emph{observed} by tests or actions.
We check this by a \emph{taint analysis}:
The inserted $v$ starts tainted; 
an assignment $y \leftarrow g$ taints $y$ if $g$ reads a
tainted variable and cleans $y$ otherwise; 
actions and tests are the observation points and must read no tainted variable.

\begin{definition}[Taint safety]
\deflabel{taint-safe}
For a taint set $T \subseteq \scvars$ and a scenario $s$, the predicate $\mathrm{safe}(T, s)$,
together with the output taint set $\mathrm{out}(T, s)$, is defined by recursion on $s$:
\begin{itemize}[nosep]
    \item an action or a test $u$: $\mathrm{safe}(T, u)$ iff $u$ reads no variable in $T$;
    $\mathrm{out}(T, u) = T$;
    \item an assignment $y \leftarrow g$: always safe;
    $\mathrm{out}(T, y \leftarrow g) = T \cup \{y\}$ if $g$ reads a variable in $T$,
    and $T \setminus \{y\}$ otherwise;
    \item $\mathtt{havoc}(y)$: always safe; $\mathrm{out}(T, \mathtt{havoc}(y)) = T \setminus \{y\}$;
    \item $\mathbf{1}$, $\mathbf{0}$: always safe, with output taint $T$;
    \item $s_1 \cdot s_2$: $\mathrm{safe}(T, s_1)$ and $\mathrm{safe}(\mathrm{out}(T, s_1), s_2)$,
    with output taint $\mathrm{out}(\mathrm{out}(T, s_1), s_2)$;
    \item $s_1 + s_2$: safe on both branches, with output taint
    $\mathrm{out}(T, s_1) \cup \mathrm{out}(T, s_2)$;
    \item $s_1^{*}$: let $T^{*}$ be the least set with $T \subseteq T^{*}$ and
    $\mathrm{out}(T^{*}, s_1) \subseteq T^{*}$, which exists since $\mathrm{out}$ is monotone
    over the finite subsets of $\scvars$; then $\mathrm{safe}(T, s_1^{*})$ iff
    $\mathrm{safe}(T^{*}, s_1)$, with output taint $T^{*}$.
\end{itemize}
\end{definition}

If $\mathrm{safe}(\{v\}, s)$, 
then $v$ never affects tests or actions in $s$, 
so we can insert a ghost assignment $v \leftarrow f$ before $s$ without changing the feasible runs it generates.
\[
    \textsc{Assign-Intro}\quad
    \frac{\mathrm{safe}(\{v\}, \mathit{post})}
         {\mathit{pre} \cdot \mathit{post} \;\incfeas\; \mathit{pre} \cdot (v \leftarrow f) \cdot \mathit{post}}
\]

\textsc{Cond-Intro} splits a run on a state predicate $r$, introducing the matching test so that a subsequent case analysis can proceed on each branch.
This rule is sound because the two branches cover all possibilities, and enjoys congruence as it does not rely on the entry state:
\[
    \textsc{Cond-Intro}\quad
    \frac{}{\ s \;\equiv_f\; (\assm(r) +\assm(\neg r))\cdot s\ }
\]

\textsc{Action-Equiv.} 
When proving $s\cdot a \inctop s$, the appended action $a$ names \emph{concrete} parameters, 
whereas the scenario $s$ describes its actions through scenario variables.
In ring leader election, for instance, the appended $\mathtt{Init}(k)$ acts at the concrete node $k$, 
while the scenario's term $\mathtt{Init}(i)$ takes its node from the loop variable $i$.
The two are not syntactically equal, 
but at the entry state where the $\mathtt{For}$ has reached iteration $k$ (by $\mathtt{For}$ unfolding, we have $i = k$), 
they resolve to the same move.
\textsc{Action-Equiv} lets us replace one by the other,
when the entry state satisfies the precondition that equates the concrete and scenario parameters:
\[
    \textsc{Action-Equiv}\quad
    \frac{(\sigma_0, c_0)\models\mathrm{wp}_{pre}\bigl(\lambda (\sigma, c).\ \evalx{e_n}{\sigma, c} = \evalx{e_n'}{\sigma, c} \wedge \evalx{e_m}{\sigma, c} = \evalx{e_m'}{\sigma, c}\bigr)}
         { pre\cdot a(e_n, e_m)\cdot post \incfeas pre\cdot a(e_n', e_m')\cdot post}
\]

\begin{lemma}[Soundness of Term-manipulation rules]
    \lemlabel{soundness-term-manipulation}
    For two scenarios $s_1, s_2$, if $s_1 \incfeas s_2$ is derivable by the term-manipulation rules, then $\semfeas{s_1} \subseteq \semfeas{s_2}$.
    By \lemref{lifting}, we have $\sem{s_1} \subseteq \sem{s_2}$.
\end{lemma}

\myparagraph{Commutativity rule (\textsc{Commute})}
The last rule realigns terms in the left-hand side into the order of the right-hand side:
it commutes a single term $t$ leftward across a whole scenario $s$, within an explicit context,
\[
    \mathit{pre} \cdot s \cdot t \cdot \mathit{post}
    \;\inccom\;
    \mathit{pre} \cdot t \cdot s \cdot \mathit{post}.
\]
Here, we generalize semi-commutativity at the level of \emph{command terms} on top of $\theta_I$ that is defined on \emph{actions}, i.e.,
$t$ can be any term, including a protocol action, assignment, or test.
Moreover, we do not explicitly interpret $s$ as action sequences and purely rely on action-level independence judgments $\theta_I$, but check semi-commutation on the syntax of $s$ directly.
The semi-independence of two command terms either comes from the protocol's semi-commutation $\theta_I$ (for two protocol actions) or from the fact that one term does not affect the other (for an action and an assignment, or a test and an assignment, etc.).
Since terms read both the scenario state and the configuration, while assignments write the former and actions the latter,
independence requires that each term's reads survive the other's write.
For a command term $u$ and a joint state $(\sigma, c)$, we write $(\sigma, c)[u]$ for the joint state after $u$:
$(\sigma[v \mapsto \evalx{e}{\sigma, c}], c)$ for an assignment $u = (v \leftarrow e)$,
$(\sigma, c')$ with $c \xrightarrow{u_{\sigma, c}} c'$ for a protocol action,
and $(\sigma, c)$ itself for a test.

\begin{definition}[Term independence]
\deflabel{term-indep}
For command terms $u, t \in \Sigma$ and a joint state $(\sigma, c)$, $t$ is \emph{independent}
of $u$ at $(\sigma, c)$, written $\mathrm{indep}_{\sigma, c}(u, t)$, by case analysis on the kinds:
\begin{itemize}[nosep]
    \item two protocol actions: each resolves the same after the other fires,
    $t_{(\sigma, c)[u]} = t_{\sigma, c}$ and $u_{(\sigma, c)[t]} = u_{\sigma, c}$,
    and $(u_{\sigma, c},\, t_{\sigma, c}) \in \theta_I$;
    \item a protocol action $a$ and an assignment $v \leftarrow e$ (in either order): the
    resolution of $a$ is unaffected by the assignment, $a_{(\sigma, c)[v \leftarrow e]} = a_{\sigma, c}$,
    and the assigned value is unaffected by the action, $\evalx{e}{(\sigma, c)[a]} = \evalx{e}{\sigma, c}$;
    \item a protocol action $a$ and a test $\assm(b)$ (in either order): the truth of $b$ is
    unaffected by the action, $(\sigma, c) \models b \Leftrightarrow (\sigma, c)[a] \models b$;
    \item a test $\assm(b)$ and an assignment $v \leftarrow e$ (in either order): the truth
    of $b$ is unaffected by the assignment, $(\sigma, c) \models b \Leftrightarrow (\sigma, c)[v \leftarrow e] \models b$;
    \item two assignments: the updates commute, $(\sigma, c)[u][t] = (\sigma, c)[t][u]$;
    \item two tests: always independent.
\end{itemize}
Every condition that mentions a protocol action is required only when that action is enabled:
a run through a disabled action is removed by the feasibility cut, so nothing is owed about it.
\end{definition}

Independence then lifts from terms to scenarios by an \emph{inductively defined} judgment
$\{\psi\}\ t \bowtie s$, read ``$t$ commutes with $s$ under $\psi$'', following the
structure of $s$ by $\mathrm{wp}$:
\[
    \frac{\forall (\sigma, c) \models \psi.\ \mathrm{indep}_{\sigma, c}(u, t)}
         {\{\psi\}\ t \bowtie u}\ \;(u \in \Sigma)
    \qquad\quad
    \frac{}{\{\psi\}\ t \bowtie \mathbf{1}}
    \qquad\quad
    \frac{}{\{\psi\}\ t \bowtie \mathbf{0}}
\]
\[
    \frac{\forall d \in \dom_v.\ \{\psi\}\ t \bowtie (v \leftarrow d)}
         {\{\psi\}\ t \bowtie \mathtt{havoc}(v)}
\]
\[
    \frac{\{\psi\}\ t \bowtie s_1,\; \{\psi'\}\ t \bowtie s_2,\;\psi \Rightarrow \mathrm{wp}_{s_1}(\psi')}
         {\{\psi\}\ t \bowtie s_1 \cdot s_2}\ \;
    \qquad
    \frac{\{\psi\}\ t \bowtie s_1,\; \{\psi\}\ t \bowtie s_2}
         {\{\psi\}\ t \bowtie s_1 + s_2}
    \qquad
    \frac{\{\psi\}\ t \bowtie s_1,\;\psi \Rightarrow \mathrm{wp}_{s_1}(\psi)}
         {\{\psi\}\ t \bowtie s_1^{*}}\
\]

On the scenario level, semi-commutativity is leveraged into contextual reasoning (similar to~\cite{farzan-safety}): an action term $\mathtt{Init}(i)$ can only be commuted with $\mathtt{Init}(j)$ when $i \ne j$, where the concrete actions rely on the scenario variables $i$ and $j$.

The commutativity rule is then a single guarded rewrite: $t$ commutes with the scenario $s$.
\[
    \textsc{Commute}\quad
    \frac{(\sigma_0, c_0) \models \mathrm{wp}_{\mathit{pre}}(\psi)
          \qquad
          \{\psi\}\ t \bowtie s}
         {\mathit{pre} \cdot s \cdot t \cdot \mathit{post}
          \;\inccom\;
          \mathit{pre} \cdot t \cdot s \cdot \mathit{post}}
\]

\begin{lemma}[Soundness of the commutativity rule]
\lemlabel{soundness-comm}
If $\mathit{pre} \cdot s \cdot t \cdot \mathit{post} \inccom \mathit{pre} \cdot t \cdot s \cdot \mathit{post}$
is derivable by \textsc{Comm}, then
$\sem{\mathit{pre} \cdot s \cdot t \cdot \mathit{post}} \subseteq \sem{\mathit{pre} \cdot t \cdot s \cdot \mathit{post}}$.
\end{lemma}


\begin{lemma}[Soundness of the calculus]
    Given two scenarios $s_1, s_2$, if $s_1 \inctop s_2$ is derivable by the calculus above, then $\sem{s_1} \subseteq \sem{s_2}$.
\end{lemma}

\begin{example}[Ring leader election: coverage]
\exlabel{rle-coverage}
By \lemref{closedness-reduction}, coverage is the base obligation $\sem{\mathbf{1}} \subseteq \sem{S}$, together
with one absorption obligation $S \cdot a \inctop S$ per action, where $S \triangleq \mathtt{scenario}_{\mathtt{RLE}}$.
The former is trivial, as every $\mathtt{PhaseChain}$ contains the empty prefix $\mathbf{1}$.
We discharge the two action types by rewriting the left-hand side down to $S$ using the calculus above.

\emph{Appending $\mathtt{Init}(k)$.}
Write $S_{<k} \triangleq \mathtt{For}(i\in[\mathtt{min},\mathtt{prev}(k)])\,\mathtt{voteChain}(i)$
and $S_{>k} \triangleq \mathtt{For}(i\in[\mathtt{next}(k),\mathtt{max}])\,\mathtt{voteChain}(i)$ for the chains below and above $k$.
We peel the outer $\mathtt{For}$ at $k$ and commute the appended action leftward across the tail iterations, 
which by \exref{rle-theta} carry only ids above $k$ and are hence $\theta$-independent of $\mathtt{Init}(k)$:
\[
\begin{aligned}
    S \cdot \mathtt{Init}(k)
    &\inctop\; S_{<k} \cdot (i \leftarrow k)\cdot\mathtt{voteChain}(k) \cdot S_{>k} \cdot \mathtt{Init}(k) && ({\rulename{Unfold-For}})\\
    &\inctop\; S_{<k} \cdot (i \leftarrow k)\cdot\mathtt{voteChain}(k) \cdot \mathtt{Init}(k) \cdot S_{>k} && ({\rulename{Commute}})
\end{aligned}
\]
Now $\mathtt{Init}(k)$ sits next to its own chain, and it suffices to absorb it there.
Unfolding the $\mathtt{PhaseChain}$ into $\mathbf{1} + P_1 + P_1\cdot P_2 + P_1\cdot P_2\cdot P_3$ and distributing $\cdot$ over $+$, 
the four prefixes give four obligations.
On the empty prefix the appended action becomes $P_1$:
\[
\begin{aligned}
    S_{<k}\cdot(i\leftarrow k)\cdot\mathbf{1}\cdot\mathtt{Init}(k)\cdot S_{>k}
    &\inctop\; S_{<k}\cdot(i\leftarrow k)\cdot\mathtt{Init}(i)\cdot S_{>k} && ({\rulename{Unit}, \rulename{Action-Equiv}})\\
    &\inctop\; S_{<k}\cdot(i\leftarrow k)\cdot\mathtt{Init}(i)\cdot(j\leftarrow i)\cdot S_{>k} && ({\rulename{Assign-Intro}})\\
    &\inctop\; S_{<k}\cdot(i\leftarrow k)\cdot P_1\cdot S_{>k}\;\inctop\; S && ({\rulename{Embed}})
\end{aligned}
\]
where \rulename{Action-Equiv} identifies the concrete $\mathtt{Init}(k)$ with the term $\mathtt{Init}(i)$ under the entry condition $i = k$, 
and \rulename{Assign-Intro} inserts the benign $(j\leftarrow i)$, which is dead and thus, unobservable in $S_{>k}$.
On the three nonempty prefixes,
$\mathtt{Init}(k)$ is disabled by the $\mathtt{Init}(i)$ already in $P_1$,
so \rulename{Disabled} rewrites each to $\mathbf 0$,
which is a subset of any $S$ by \rulename{Bot}.

\emph{Appending $\mathtt{Recv}(k, m)$.}
Similarly as with $\mathtt{Init}$, we write $v \triangleq m.\mathtt{id}$, peel the $\mathtt{For}$ at $v$, and commute
the appended receive across $S_{>v}$ (whose actions concern votes above $v$, hence
$\theta$-independent), reaching the context $C[\,\cdot\,] = S_{<v} \cdot [\,\cdot\,] \cdot S_{>v}$ with
the receive placed just after $\mathtt{voteChain}(v)$.
Using \rulename{Action-Equiv} to identify the concrete $m$ with the scenario's computed term
$b \triangleq \mathtt{Recv}(\mathtt{right}(j),\mathtt{id}(i))$, and distributing the $\mathtt{PhaseChain}$ over its four prefixes $\mathbf 1$, $P_1$, $P_1\cdot P_2$, $P_1\cdot P_2\cdot P_3$, the receive meets each prefix in turn.

On $\mathbf 1$ no vote is in flight, and on $P_1\cdot P_2\cdot P_3$ all votes on $i$ are already received. 
On both the receive is disabled, so \rulename{Disabled} prunes them.

On the $P_1$ prefix the vote has just been initiated, so the receive is its \emph{first} relay.
Recall $P_2 = \bigl(\assm(\neg\mathtt{isMax})\cdot b\cdot(j\leftarrow\mathtt{right}(j))\bigr)^{*}$,
and Kleene algebra gives us $1 \incSig s^{*}$ and $s \incSig s^{*}$.
We split on the successor and fold each branch back:
\[
\begin{aligned}
    C[P_1\cdot b]
    \;\inctop\;\ & C[P_1\cdot \assm(\neg\mathtt{isMax})\cdot b\cdot(j\leftarrow\mathtt{right}(j))] \\
    &{}+\, C[P_1\cdot \assm(\mathtt{isMax})\cdot b]
    && ({\rulename{Cond-Intro}, \rulename{Assign-Intro}})\\
    \;\inctop\;\ & C[P_1\cdot \bigl(\assm(\neg\mathtt{isMax})\cdot b\cdot(j\leftarrow\mathtt{right}(j))\bigr)^{*}] \\
    &{}+\, C[P_1\cdot P_3]
    && ({\rulename{Unfold}}, \rulename{Embed})\\
    =\ & C[P_1\cdot P_2] \,+\, C[P_1\cdot P_2 \cdot P_3]
\end{aligned}
\]

The $P_1\cdot P_2$ prefix is identical, except the vote has already relayed some hops, so the not-max
branch is \emph{one more} iteration appended to the existing star $P_2 = \bigl(\assm(\neg\mathtt{isMax})\cdot b\cdot(j\leftarrow\mathtt{right}(j))\bigr)^{*}$,
where we have $s^{*}\cdot s \incSig s^{*}$ from Kleene algebra.
\[
\begin{aligned}
    C[P_1\cdot P_2 \cdot b]
    \;\inctop\;\ & C[P_1\cdot P_2 \cdot \assm(\neg\mathtt{isMax})\cdot b\cdot(j\leftarrow\mathtt{right}(j))] \\
    &{}+\, C[P_1\cdot P_2 \cdot \assm(\mathtt{isMax})\cdot b]
    && ({\rulename{Cond-Intro}, \rulename{Assign-Intro}})\\
    \;\inctop\;\ & C[P_1\cdot P_2] \,+\, C[P_1\cdot P_2\cdot P_3]
    && (\rulename{Unfold},\ \rulename{Embed})
\end{aligned}
\]
Every summand lies in $C[\mathtt{voteChain}(v)] = S$, so the run rewrites back into $S$, giving $S\cdot\mathtt{Recv}(k,m)\inctop S$.
This completes the coverage proof for ring leader election.
\end{example}

\subsection{Completeness}
\seclabel{completeness}

At the top level the system is complete relative to the underlying assertion logic:
every safety property that has an inductive invariant expressible in the assertion logic has a scenario proof.
The witness is a maximally permissive scenario that generates every instantiated action with $\mathtt{havoc}$.

\myparagraph{The maximally permissive scenario}
Let $\acts = \{a_1, \dots, a_k\}$.
For each declaration $a_i$, fix two fresh scenario variables: $x_i$ of node type $\ntype(a_i)$ and $y_i$ of message type $\mtype(a_i)$
(we drop $y_i$ when $\mtype(a_i) = \bot$).
Define
\[
    s_{a_i} = \mathtt{havoc}(x_i) \cdot \mathtt{havoc}(y_i) \cdot a_i(x_i, y_i),
    \qquad
    s_{\top} = (s_{a_1} + s_{a_2} + \cdots + s_{a_k})^{*}.
\]
Since $\mathtt{havoc}$ ranges over every value of the variable's type,
$\semvar{s_{a_i}} = \{(a_i, n, m) \mid n \in \nodes_{\ntype(a_i)},\ m \in \msgs_{\mtype(a_i)} \cup \{\bot\}\}$.
The union over the finitely many declarations is exactly $\actsI$, so $\semvar{s_{\top}} = \actsI^{*}$ and $\semfeas{s_{\top}} = \fexec_I$.
We take $\theta_I = \emptyset$, which is trivially sound, so that $\sem{s_{\top}} = \fexec_I$.
No assumption on the value domain is needed: $\mathtt{havoc}$ does not enumerate values,
so parameters ranging over an infinite (even uncountable) payload domain are generated as well.

\myparagraph{Coverage is derivable}
Coverage of $s_{\top}$ holds semantically, and it is also derivable in the calculus of \secref{coverage}, using only the absorb pattern.
The base obligation $\mathbf{1} \inctop s_{\top}$ follows from \textsc{Unfold} and \textsc{Embed}.
For an appended action $\alpha = (a_i, n, m)$, written as the term $a_i(n, m)$ with constant parameters,
\[
\begin{aligned}
    s_{\top} \cdot a_i(n, m)
    &\inctop\; s_{\top} \cdot (x_i \leftarrow n) \cdot (y_i \leftarrow m) \cdot a_i(n, m) && (\rulename{Assign-Intro}\ \text{twice})\\
    &\inctop\; s_{\top} \cdot (x_i \leftarrow n) \cdot (y_i \leftarrow m) \cdot a_i(x_i, y_i) && (\rulename{Action-Equiv})\\
    &\inctop\; s_{\top} \cdot s_{a_i} \;\inctop\; s_{\top} && (\rulename{Havoc},\ \rulename{Embed},\ \rulename{Unfold})
\end{aligned}
\]
\rulename{Assign-Intro} applies because the only remaining term, $a_i(n, m)$, reads no variable.
The premise of \rulename{Action-Equiv} is $\mathrm{wp}$ of $x_i = n \wedge y_i = m$ through the two assignments, which is $\mathsf{true}$,
and so is its $\mathrm{wp}$ through $s_{\top}$, with loop invariant $\mathsf{true}$.

\myparagraph{Correctness is an inductive invariant}
Correctness of $s_{\top}$ reduces to its loop invariant.
Take $J = \lambda (\sigma, c).\ \mathit{Inv}(c)$ for a configuration predicate $\mathit{Inv}$.
By the $\mathrm{wp}$ rules,
\[
    \mathrm{wp}_{s_{a_i}}(J)(\sigma, c)
    \;=\;
    \forall n, m.\ \mathrm{en}_{(a_i, n, m)}(c) \Rightarrow \mathit{Inv}(c')
    \quad \text{where } c \xrightarrow{(a_i, n, m)} c',
\]
so the loop obligation $J \Rightarrow \mathrm{wp}_{s_{a_1} + \cdots + s_{a_k}}(J)$ states exactly that $\mathit{Inv}$ is preserved by every enabled instantiated action.
The remaining obligations are $\mathit{Inv}(c_0)$ and $\mathit{Inv} \Rightarrow \varphi$.
Hence any inductive invariant $\mathit{Inv}$ establishing $\varphi$ yields a scenario proof on $s_{\top}$.
Conversely, if $\protocol \models \varphi$, then the reachable configurations form such an invariant in every instance.

\begin{theorem}[Relative completeness]
\thmlabel{relative-completeness}
Suppose the assertion logic is closed under the $\mathrm{wp}$ operators,
in particular under universal quantification over node and message sorts, which $\mathrm{wp}_{\mathtt{havoc}}$ introduces.
If $\varphi$ has an inductive invariant expressible in the assertion logic,
then $\protocol \models \varphi$ is provable by a scenario proof using $s_{\top}$ and $\theta_I = \emptyset$.
\end{theorem}

The classical inductive-invariant method is therefore the special case of a scenario proof on $s_{\top}$.
Both its premises are about the assertion logic: the closure condition, for $\mathrm{wp}_{\mathtt{havoc}}$,
and expressiveness, for the invariant.
The scenario language itself imposes no further condition, because $\mathtt{havoc}$ generates every action.
The complexity of inductive-invariant reasoning comes from the high-degree of nondeterminism in $s_{\top}$.
A natural scenario reduces the amount of nondeterminism by structuring the actions into a more sequential manner: local reasoning along causal histories, and a much smaller number of interleaving patterns due to semi-commutativity.

\begin{remark}[Incompleteness of the coverage calculus]
\remlabel{coverage-incomplete}
A separate question is the completeness of the coverage calculus, and here the picture is less settled.
The rules of \secref{coverage} are sound but \emph{far from complete} 
as a decision procedure for $s_1 \inctop s_2$.
For instance, given a predicate $r$, 
we have $\sem{\assm(r)\cdot s} \;\subseteq\; \sem{\assm(r) \cdot \assm(r)\cdot s}$.
However, the calculus has no Boolean algebra rules to discharge the reasoning over the tests $\assm(r)$.
This gap does not surface in our case studies because they use only the \emph{absorb} pattern: 
the obligation $s \cdot a \inctop s$ has a left- and a right-hand side that are almost identical,
rather than compared against an arbitrarily different scenario.
Moreover, inclusion of commutation closures of regular languages is co-recursively enumerable (a non-inclusion has a finite word witness) and undecidable~\cite{trace-equivalence-undecidable}, hence not recursively enumerable.
It embeds into $\sem{s_1} \subseteq \sem{s_2}$ via always-enabled, no-effect actions,
so no sound calculus whose side conditions are effectively checkable can be complete in general.
\end{remark}

\section{Case Study: Paxos}
\seclabel{casestudy}

Ring leader election has served as the running example of \secref{scenario} and
\secref{proof-system}. 
This section evaluates the methodology on single-decree Paxos.
Due to the page limit, scenario proofs of further protocols from the Ivy benchmark~\cite{ivy}
(two-phase commit, decentralized lock, two variants of multi-signature payment) are presented in the Appendix.

\myparagraph{Protocol}
Single-decree Paxos~\cite{paxos} is the archetypal consensus protocol: 
competing proposers race to get a quorum of acceptors to agree on a single value.
Each proposer $p$ owns a unique round, its identifier; $n$ is the number of acceptors.
The actions are as follows; every receipt consumes its message:
\begin{itemize}[nosep,leftmargin=*]
    \item $\mathtt{Prepare}(p)$: $p$ broadcasts a prepare message to every acceptor, once.
    \item $\mathtt{Promise^{\pm}}(a, m)$: on a prepare from $p$, acceptor $a$ promises ($+$) if $p > a.\mathtt{maxPromisedRound}$,
    replying with its accepted pair $(a.\mathtt{accRound}, a.\mathtt{accValue})$ and setting $a.\mathtt{maxPromisedRound} \leftarrow p$;
    otherwise it ignores $m$ ($-$).
    \item $\mathtt{RecvPromise}(p, m)$: $p$ adds the sender to $p.\mathtt{promisedSet}$ and keeps the reported pair of highest round in $(p.\mathtt{maxRound}, p.\mathtt{maxValue})$.
    \item $\mathtt{Propose}(p)$: once $2 \cdot |p.\mathtt{promisedSet}| > n$, $p$ sets $p.\mathtt{value}$ to $p.\mathtt{maxValue}$ unless it is $\bot$,
    and broadcasts a proposal carrying $p.\mathtt{value}$, once.
    \item $\mathtt{Accept^{\pm}}(a, m)$: on a proposal $(p, v)$, $a$ accepts ($+$) if $p \ge a.\mathtt{maxPromisedRound}$,
    setting $a.\mathtt{maxPromisedRound} \leftarrow p$, $(a.\mathtt{accRound}, a.\mathtt{accValue}) \leftarrow (p, v)$, and adding $(p, v)$ to $a.\mathtt{acceptedHistory}$;
    otherwise it ignores $m$ ($-$).
\end{itemize}

The safety property is \emph{agreement}: calling a pair $(r, v)$ \emph{chosen} when a strict majority of acceptors holds it in their accepted histories, \emph{any two chosen values are equal}.
Correctness is hard because all rounds run \emph{concurrently}: 
there is no global order in which one round completes before the next begins, 
so the messages received by a single acceptor of different rounds interleave in combinatorially many ways.
Next, we will show how to write a systematic scenario that captures all of these interleavings by sequentializing the actions of different rounds, and correctness can easily be proved along it.

\myparagraph{Scenario}
\begin{figure}[t]
\small
\hrule\smallskip
{\raggedright\textbf{Scenario: $\mathtt{scenario}_{\mathtt{Paxos}} \triangleq S_{\mathrm{pos}} \cdot S_{\mathrm{neg}}$}\par}
\smallskip

{\itshape Positive part --- each round run to completion, in ascending order:}
\begin{align*}
    S_{\mathrm{pos}} &\triangleq \mathtt{For}\bigl(p \in [\mathtt{min}, \mathtt{max}]\bigr)\;
        \mathtt{PhaseChain}[P_1, P_2, P_3, P_4, P_5] \\[3pt]
    P_1 &\triangleq \mathtt{Prepare}(p) && \text{\itshape prepare}\\[2pt]
    P_2 &\triangleq \mathtt{For}\bigl(a \in [\mathtt{min}, \mathtt{max}]\bigr)\;
        \mathtt{PhaseChain}\big[\,\mathtt{Promise^{+}}(a, \mathtt{prep}_{p,a}), && \text{\itshape promise \& receive}\\
        &\qquad \mathtt{RecvPromise}(p, \mathtt{prom}_{a,p}(a.\mathtt{accRound}, a.\mathtt{accValue}))\,\big] \\[2pt]
    P_3 &\triangleq \assm(2 \cdot |p.\mathtt{promisedSet}| > n) \cdot \mathtt{Propose}(p) && \text{\itshape propose}\\[2pt]
    P_4 &\triangleq \bigl(\mathtt{havoc}(x) \cdot \mathtt{Promise^{+}}(x, \mathtt{prep}_{p,x})\bigr)^{*} && \text{\itshape late promises}\\
        &\qquad \cdot \bigl(\mathtt{havoc}(x)\,\mathtt{havoc}(r)\,\mathtt{havoc}(v) \cdot \mathtt{RecvPromise}(p, \mathtt{prom}_{x,p}(r, v))\bigr)^{*} \\[2pt]
    P_5 &\triangleq \mathtt{For}\bigl(a \in [\mathtt{min}, \mathtt{max}]\bigr)\;
        \mathtt{PhaseChain}\bigl[\,\mathtt{Accept^{+}}(a, \mathtt{prop}_{p,a}(p.\mathtt{value}))\,\bigr] && \text{\itshape accept}
\end{align*}

\smallskip
{\itshape Negative part --- all ignored receipts, deferred to the end:}
\begin{align*}
    S_{\mathrm{neg}} &\triangleq
        \bigl(\mathtt{havoc}(x)\,\mathtt{havoc}(q) \cdot \mathtt{Promise^{-}}(x, \mathtt{prep}_{q,x})\bigr)^{*} && \text{\itshape ignored prepares}\\
        &\qquad \cdot \bigl(\mathtt{havoc}(x)\,\mathtt{havoc}(q)\,\mathtt{havoc}(v) \cdot \mathtt{Accept^{-}}(x, \mathtt{prop}_{q,x}(v))\bigr)^{*} && \text{\itshape ignored proposals}
\end{align*}
\smallskip\hrule
\caption{The scenario $\mathtt{scenario}_{\mathtt{Paxos}}$ for single-decree Paxos. $\mathtt{prep}_{p,a}$, $\mathtt{prom}_{a,p}(r, v)$, $\mathtt{prop}_{p,a}(v)$
    are the prepare, promise (reporting the accepted pair $(r, v)$), and proposal (carrying value $v$) messages.
    The scenario has no ghost state: message payloads and the quorum test read the local states.}
\figlabel{paxos-scenario}
\end{figure}

The developer's picture of a run is one causal history per round, each passing through the phases prepare, promise, propose, accept.
Given an arbitrary interleaving of the causal histories, let us reorder it into a representative run that can be described by a scenario in two steps.
First, all the ``negative'' actions, the $\mathtt{Promise^{-}}$s and $\mathtt{Accept^{-}}$s, can be deferred to the end of the run, 
because if they are enabled at some point, the $\mathrm{maxPromisedRound}$ can only increase, and they remain enabled until the end.
Now, we have a prefix of the run consisting of only the ``positive'' actions.
In the second step, we can move the actions of smaller round before those of larger round, 
because of a similar reason: if the two actions are executed by the same acceptor, 
the action of smaller round must be before the action of larger round, 
otherwise actions on different acceptors are independent and can be reordered in either way.
This observation shapes the prefix: 
the rounds are executed in ascending order, and each round runs to completion before the next begins, without any interleaving of actions from different rounds.

The resulting scenario is given in \figref{paxos-scenario}: a positive part $S_{\mathrm{pos}}$
followed by a negative part $S_{\mathrm{neg}}$.
The positive part is a $\mathtt{For}$ over rounds in ascending order, each round run to completion as
a $\mathtt{PhaseChain}$ of five phases:
$P_1$ prepares;
$P_2$ collects promises, one acceptor at a time;
$P_3$ proposes once $p$ has a quorum;
$P_4$ absorbs promises that arrive \emph{after} the proposal
(they can no longer influence the proposed value, so havoc'd stars admit them, a harmless over-approximation);
and $P_5$ accepts, one acceptor at a time.
The negative part collects every ignored receipt with havoc'd stars in the same way.

\myparagraph{Semi-commutativity}
The two reordering steps above are justified by a semi-commutation $\theta$ (\figref{paxos-theta}),
which rests on exactly the two observations that shaped the representative run:
(a) a positive action may be reordered before a negative one, and
(b) an action of a smaller round may be reordered before an action of a larger round.
We group the conditions as follows:
\begin{enumerate}
    \item \emph{Two proposer actions} ($C_1$) are independent when they belong to distinct rounds.
    \item \emph{A proposer action and an acceptor action} ($C_2$) are linked by message passing, so they are ordered only when the right one receives the message the left one sends.
    \item \emph{Two non-negative actions at one acceptor} ($C_4$) are dependent, \emph{unless} the rounds are descending. This is the case observation (b) relies on, and it is what lets an action of round $p$ cross the actions that later rounds perform.
    \item \emph{Two $\mathtt{RecvPromise}$s} ($C_3$) update the proposer's maximum learned pair, so they commute unless they reach the same proposer with the same round but different values.
    \item Finally, an \emph{negative} action as the left argument commutes with anything: it touches nothing at the acceptor, so both orders reach the same state, and it remains enabled afterwards because its guard $p \le a.\mathtt{maxPromisedRound}$ survives any action that can only raise $\mathtt{maxPromisedRound}$.
    The relation is one-directional here, $({-},{+}) \rew ({+},{-})$. This costs nothing: after the split into $S_{\mathrm{pos}}$ and $S_{\mathrm{neg}}$ the negative actions already sit at the tail of the scenario, so no commutation of that shape is ever required.      
\end{enumerate}

\begin{figure}[t]
\centering
\small
\setlength{\tabcolsep}{5pt}
\newcommand{\hd}[1]{\rotatebox[origin=l]{45}{$#1$}}
\begin{tabular}{r|ccccccc}
\multicolumn{1}{r}{} & \hd{\mathtt{Prep}(p_2)} & \hd{\mathtt{Prom^{+}}(a_2, m_2)} & \hd{\mathtt{RecvP}(p_2, m_2)} & \hd{\mathtt{Prop}(p_2)} & \hd{\mathtt{Acc^{+}}(a_2, m_2)} & \hd{\mathtt{Prom^{-}}(a_2, m_2)} & \hd{\mathtt{Acc^{-}}(a_2, m_2)} \\
\hline
$\mathtt{Prep}(p_1)$          & $C_1$ & $C_2$ & $C_1$ & $C_1$ & \tT   & \tF & \tF \\
$\mathtt{Prom^{+}}(a_1, m_1)$ & \tT   & $C_4^{\le}$ & $C_2$ & \tT   & $C_4^{<}$ & \tF & \tF \\
$\mathtt{RecvP}(p_1, m_1)$    & $C_1$ & \tT & $C_3$ & $C_1$ & \tT   & \tF & \tF \\
$\mathtt{Prop}(p_1)$          & $C_1$ & \tT   & $C_1$ & $C_1$ & $C_2$ & \tF & \tF \\
$\mathtt{Acc^{+}}(a_1, m_1)$  & \tT   & $C_4^{\le}$ & \tT   & \tT & $C_4^{<}$ & \tF & \tF \\
$\mathtt{Prom^{-}}(a_1, m_1)$ & \tT & \tT & \tT & \tT & \tT & \tT & \tT \\
$\mathtt{Acc^{-}}(a_1, m_1)$  & \tT & \tT & \tT & \tT & \tT & \tT & \tT \\
\end{tabular}

\medskip
\begin{minipage}{\linewidth}
\centering
$C_1 \;\equiv\; p_1 \ne p_2$
\qquad
$C_4^{\lhd} \;\equiv\; a_1 \ne a_2 \;\vee\; m_2.\sender \lhd m_1.\sender$
\quad $({\lhd} \in \{{\le}, {<}\})$

\smallskip
$C_3 \;\equiv\; p_1 \ne p_2 \;\vee\; (m_1.\mathtt{accRound} = m_2.\mathtt{accRound} \Rightarrow m_1.\mathtt{accValue} = m_2.\mathtt{accValue})$
\end{minipage}
\caption{Semi-commutation relation $\theta$ for Paxos, read row-then-column: each cell gives the
condition under which the ordered pair is in $\theta$ (\tT{} always, \tF{} never).
$C_2$: the right action does not receive the message the left one sends.
Pairs whose forward order is infeasible (e.g., a $\mathtt{Prom^{+}}$ before its $\mathtt{Prep}$) are admitted vacuously.}
\figlabel{paxos-theta}
\end{figure}

\myparagraph{Coverage}
The key observation follows the intuition of the scenario: every ignored action is absorbed by the negative part at the end, and every other action of a round can commute leftward across the negative part and all rounds above it, until it is absorbed into the body of its own round.
Writing $S = S_{\mathrm{pos}}\cdot S_{\mathrm{neg}}$, $S_{\mathrm{neg}} = N_{\mathrm{P}}^{*} \cdot N_{\mathrm{A}}^{*}$, and, for an appended action $a$ of round $p$, splitting $S_{\mathrm{pos}}$ at round $p$ into $R_{<p}\cdot\mathtt{round}(p)\cdot R_{>p}$, the two cases run as follows.
A \emph{negative} $a$ (an ignored receipt) is absorbed by its own star, after a $\mathtt{Promise^{-}}$ commutes leftward across $N_{\mathrm{A}}^{*}$:
\[
\begin{aligned}
    S \cdot \mathtt{Promise^{-}}(\cdot)
    &\;\inctop\; S_{\mathrm{pos}} \cdot N_{\mathrm{P}}^{*} \cdot \mathtt{Promise^{-}}(\cdot) \cdot N_{\mathrm{A}}^{*}
    && ({\rulename{Commute}})\\
    &\;\inctop\; S_{\mathrm{pos}} \cdot N_{\mathrm{P}}^{*} \cdot N_{\mathrm{P}} \cdot N_{\mathrm{A}}^{*} \;\inctop\; S
    && (\rulename{Assign-Intro}, \rulename{Action-Equiv}, \rulename{Havoc}, \rulename{Unfold})
\end{aligned}
\]
A \emph{non-negative} $a$ commutes leftward across $S_{\mathrm{neg}}$ and the later rounds, to the body of its own round:
\[
\begin{aligned}
    S \cdot a
    &\;\inctop\; S_{\mathrm{pos}} \cdot a \cdot S_{\mathrm{neg}}
    && ({\rulename{Commute}})\\
    &\;\inctop\; R_{<p} \cdot \mathtt{round}(p) \cdot R_{>p} \cdot a \cdot S_{\mathrm{neg}}
    && ({\rulename{Unfold-For}})\\
    &\;\inctop\; R_{<p} \cdot \mathtt{round}(p) \cdot a \cdot R_{>p} \cdot S_{\mathrm{neg}}
    && ({\rulename{Commute}})
\end{aligned}
\]
Now the remaining task is to show that the round on $p$ can absorb the appended action,
which is easy because the round is written sequentially and \emph{closed} with respect to a single round, and the appended action is one of its own.
We illustrate with $\mathtt{Propose}(p)$.
Unfolding the round body,
$\mathtt{round}(p) = \mathbf 1 + P_1 + P_1\cdot P_2 + P_1\cdot P_2\cdot P_3 + P_1\cdot P_2\cdot P_3\cdot P_4 + P_1\cdot P_2\cdot P_3\cdot P_4\cdot P_5$,
and distributing the appended action over the six prefixes, all but one are pruned:
on $\mathbf 1$ and $P_1$ no quorum has been collected, and on the prefixes past $P_3$ the proposer has already proposed ($\mathtt{isProposed}$), 
so the guard of $\mathtt{Propose}(p)$ fails and \rulename{Disabled} discards these branches.
The only surviving prefix is $P_1\cdot P_2$, where the appended action becomes the phase $P_3$.
Writing $C[\,\cdot\,] = R_{<p} \cdot P_1\cdot P_2 \cdot [\,\cdot\,] \cdot R_{>p} \cdot S_{\mathrm{neg}}$ for the surrounding context
and $q \equiv 2 \cdot |p.\mathtt{promisedSet}| > n$ for the quorum test:
\[
\begin{aligned}
    C[\mathtt{Propose}(p)]
    &\;\inctop\; C\bigl[\bigl(\assm(q) + \assm(\neg q)\bigr)\cdot \mathtt{Propose}(p)\bigr]
    && ({\rulename{Cond-Intro}})\\
    &\;\inctop\; C\bigl[\assm(q)\cdot \mathtt{Propose}(p)\bigr]
    && ({\rulename{Disabled}})\\
    &\;=\; C[P_3] \;\inctop\; S
    && ({\rulename{Embed}})
\end{aligned}
\]
\rulename{Cond-Intro} splits on the quorum test, and \rulename{Disabled} prunes the $\assm(\neg q)$
branch.
The other actions absorb the same way, each positioned in the phase it belongs to.

\myparagraph{Correctness}
The correctness on the sequentialized scenario is straightforward as a round finishes before the next begins (serves as the loop invariant over rounds):
if a value $v$ is chosen in round $i$ in a quorum, then all the later rounds will always propose $v$ and potentially accept $v$.
Consider a round $p > i$ issuing its proposal in $P_3$.
Since round $i$ ran to completion before $p$ began, its accepting quorum is already in place, and
$p$'s own promise quorum must intersect it in some acceptor $\hat a$.
The promise $\hat a$ sends therefore reports a pair of round at least $i$ (which must be $v$ by the invariant), 
so $p.\mathtt{maxValue}$ is set to $v$ in $P_2$ and $p$ proposes $v$ in $P_3$.

In the invariant-based reasoning~\cite{paxos-epr}, the crucial inductive invariant says that when the owner of round $i$ proposes, it cannot miss any value that is still \emph{choosable} at a lower round in the future, 
where ``choosable'' predicts what some future execution might do.
Such a statement is hard to phrase even informally, let alone in logic. 
However, in our case, thanks to commutativity, the complexity is totally avoided by the fact that no reasoning about overlapping rounds is needed: 
because each round is checked against a state in which all earlier rounds have finished, 
and the combinatorial interleavings that make the invariant-based proof of Paxos hard never arise.

\section{Related Work}
\seclabel{related}

\myparagraph{Scenario-based proofs for concurrent objects}
While the high level idea of expressing a protocol through scenarios and exploiting commutativity
quotients has indeed been inspired from~\cite{enea-koskinen}, the details about how the expression
is used to establish correctness differ in significant ways.
First, our formalism is carefully crafted so that we systematically get a completeness guarantee.
On the other hand, the work of~\cite{enea-koskinen} does not address this.
It instead proposes an automata-based approach, and only works when the scenario can be expressed using \emph{layers}, 
each of which consists of a single shared-memory write with local reads and their reactions around it.
In this fragment coverage reduces to two-thread reasoning and candidate expressions can be derived automatically from source code.
Distributed protocols rarely fit this shape, since a protocol's causal history is threaded by message passing and is in general a DAG rather than a single write with reactions around it.
For this reason, we structure our proof system so that coverage is actually an explicit obligation over unrestricted scenarios, 
and further propose a syntax-directed calculus to discharge it.

\myparagraph{Representative-execution proofs via commutativity}
The idea of proving safety by exhibiting a canonical family of executions, with a commutativity argument that they cover all behavior, is shared by several lines.
The closest work is \emph{inductive sequentialization} (IS)~\cite{inductive-seq}, which also proposes the use
of sequential runs to represent behavior of protocols.
On Paxos, its sequentialization is the same intuition as ours, rounds in increasing order.
Nonetheless, modeling protocols as \emph{asynchronous programs} rather than as actions, IS constructs a \emph{transformation}, repeatedly absorbing a pending asynchronous call until what remains is sequential, where we construct a \emph{description}, a scenario expression we prove covers every feasible run, replacing an inductive elimination of pending calls by a syntax-directed calculus deriving inclusions between expressions.
The consequence is where commutativity is demanded: absorbing an action into a whole asynchronous program, IS requires it to commute against every action, which real protocol actions often violate, forcing the user to abstract each action into a mover.
On an already sequential scenario, a term commutes only to the part it actually crosses. 
An ignored receipt in Paxos, for instance, is absorbed into the negative part at the end and never meets the rounds, so no independence between the two is needed.
\emph{Canonical sequentialization}~\cite{canonical-seq} computes one automatic form for a symmetric fragment that excludes ring leader election and Paxos;
\emph{pretend synchrony}~\cite{pretend-synchrony} and communication-closed reductions~\cite{cc-async} fix a synchronous round form.
This line descends from Lipton's movers~\cite{lipton} through QED~\cite{qed}, the \textsc{Civl}
layered-refinement tradition~\cite{civl,synchronizing,civl-verifier}, to inductive sequentialization~\cite{inductive-seq}.

\myparagraph{Invariant-based verification of distributed protocols}
The dominant approach states correctness as an inductive invariant or invariant-carrying refinement:
Ivy and its EPR reasoning~\cite{ivy,paxos-epr,mypyvy}, the Lean-embedded Veil~\cite{veil},
TLA${}^{+}$/TLAPS~\cite{tlaplus,tlaps}, IronFleet~\cite{ironfleet}, Verdi~\cite{verdi}, 
and the separation logics Disel, Igloo, and Aneris~\cite{disel,igloo,aneris}.
A large body automates the invariant~\cite{swiss,distai,duoai,fol-ic3,induction-duality,i4} or restructures it~\cite{kondo,basilisk,ado,grove,trillium}; 
both are orthogonal to our change of artifact.
\emph{Message chains}~\cite{message-chains} also falls under this umbrella, and merely extends the same style of reasoning so that invariants may quantify over the causal history of a message.
As our completeness observation shows, the invariant method is the degenerate scenario proof on $s_\top$.
A natural scenario usually reduces the invariant complexity by sequentializing the causal histories and picking representative executions, thereby reducing the degree of nondeterminism.

\myparagraph{Algorithmic commutativity}
A parallel line \emph{searches} for a reduction and its proof together rather than accepting a written one~\cite{farzan-stratified,farzan-parameterized,farzan-hypersafety,farzan-sound-seq}.
In particular, \emph{reductions for safety proofs}~\cite{farzan-safety} introduces contextual commutativity, the algorithmic cousin of our $\psi$-guarded rule.
Partial-order reduction~\cite{dpor,godefroid-por,optimal-dpor}, resting on Mazurkiewicz trace equivalence~\cite{mazurkiewicz,semi-commutations}, 
shares the same one-representative-per-class insight, 
but explores it for bounded instances rather than discharging it deductively over parameterized protocols.

\myparagraph{Algebraic foundations}
The structure layer of our calculus is Kleene algebra with tests~\cite{kozen-kat,kozen-completeness}.
The commutativity layer adds semi-commutation hypotheses~\cite{kozen-commutativity,ka-hypotheses},
whose equational theory is undecidable~\cite{kat-comm-undecidable}; in the language model, inclusion is moreover not recursively enumerable~\cite{trace-equivalence-undecidable}, so our coverage sub-system is necessarily incomplete (\remref{coverage-incomplete}).
Related domain-specific and concurrent variants include concurrent Kleene algebra~\cite{cka}, NetKAT~\cite{netkat}, and GKAT~\cite{gkat}.

\myparagraph{Scenario-based specification}
Our ``scenario'' should be distinguished from a tradition that treats scenarios as specifications or programs to realize or execute: message sequence charts and their
realizability~\cite{msc-modelcheck,msc-realizability}, live sequence charts and behavioral programming~\cite{lsc}, scenario-based synthesis~\cite{scenario-synthesis},and from work mining the diagrams engineers draw for invariants~\cite{flows}.
Reduction and refinement in the temporal-logic world is likewise related~\cite{reduction-tla,refinement-mappings,action-systems,bakery}.
None of these turns a user-written scenario into a \emph{proof} of an independently given protocol, which is the object a scenario proof provides.


\section{Conclusion}
\seclabel{conclusion}


This paper argues that a distributed protocol is better verified by
writing down an \emph{operational} description of \emph{how it runs}
than by guessing a \emph{declarative} description of \emph{what stays true}
and strengthening it to exclude unreachable counterexample states.
Towards this, we introduce scenario proofs to make this operational
reasoning precise.
The main constituent of such a proof is a \emph{scenario expression}
that describes representative executions using causal histories and phases.
This expression is accompanied by proofs of \emph{coverage}, i.e.,
that every feasible execution of the protocol is represented in the
scenario modulo semi-commutativity, and \emph{correctness}, i.e.,
that the desired safety property holds for the representative executions.
The soundness of the semi-commutation relation is established locally
on pairs of actions.
By removing unnecessary interleavings from the correctness argument,
scenarios ease the task of finding invariants and bring the structure
of the proof closer to how protocol developers organize their informal
correctness arguments.

At a conceptual level, scenario proofs as we present them here
can also be positioned close to several works in automated verification that follow the
\emph{trace abstraction} approach~\cite{Heizmann2009TraceAbstraction,NestedInterpolants2010,Farzan2014Proofs,farzan-stratified}.
In these works, the primary proof object, like in our setup, is a
succinct but explicit representation of executions of the system of interest.
Trace abstraction, in various forms, has demonstrated that this style
of reasoning can be powerful in completely automated verification:
the Ultimate (Automizer) family of tools has consistently been the winner of
the yearly software verification competition (SVCOMP) for several years now.
This was indeed a starting point of this work --- we wanted to
understand why and how this style of reasoning helps.
We do believe we have a partial answer for this, and this work is a testament to it --- execution-based
representations such as scenarios offer a unique flexibility of representation,
unlike the otherwise stiff, invariant-style, syntax-driven proofs \`a la Hoare logic
(and the myriad of other program logics in the same family).

The reframing of correctness reasoning using scenarios 
runs deeper than a new proof technique for a single safety property, and opens
up other interesting directions.
At the level of what a scenario \emph{is}, a scenario together with its coverage proof is a
property-independent characterization of a protocol's behavior, established once and reused: the same
object should discharge further safety properties at the cost of only a new correctness argument, and
we conjecture it can carry reasoning beyond safety: liveness, compositional reasoning, and
modular protocol design.
Further, because a scenario is an executable description rather than a mere assertion, it invites a design-time use: \emph{synthesizing} a protocol from the scenario a designer sketches, rather than verifying one after the fact.
In this sense, both choreographic programming techniques~\cite{ChoreographicProgramming2013}
and top-down synthesis based on session types~\cite{MPST2016} share a similar goal.
At the level of the proof system itself, two questions follow most directly.
How complete is the coverage calculus, and how much of it can be handed to decision procedures for Kleene algebra with tests, the step that would turn scenario proofs from a discipline into a tool.
Next, scenarios constitute a fairly behavioral description of systems (protocols in our case),
and this means they naturally serve as documentation for protocols
and can be used to guide testing campaigns to eliminate design-time issues~\cite{Gulcan2025ModelGuided};
indeed recent works on testing of formal models has shown this style of reasoning to be effective~\cite{LiDesaiPadhye2026}.
And, because we claim scenarios mirror the way designers already argue, 
how well a large language model can be driven to \emph{write} the scenario and its proof, turning that naturalness into automation and closing the loop back to the informal executions from which designers start.

\bibliographystyle{ACM-Reference-Format}
\bibliography{ref}
\newpage
\appendix

\section{Further Case Studies}
\applabel{more-casestudies}

Due to the page limit, this appendix presents the scenario proofs of four further protocols:
two-phase commit and a decentralized lock from the Ivy benchmark, and two variants of a
multi-signature payment protocol.

\subsection{Two-Phase Commit}
\seclabel{cs-2pc}

Two-phase commit coordinates an atomic decision --- commit or abort --- across a
coordinator and a set of participants.
Where ring leader election was a single chain replicated per node, two-phase commit exhibits a single causal history with a clear \emph{phase}
structure, a voting round bracketed by a broadcast and a decision.

\myparagraph{Protocol}
Following \secref{prelim}, \figref{tpc-protocol} gives the protocol in full: two node types, a
single coordinator $\mathtt{c}$ and $n$ participants; five message types; and seven actions.
The coordinator records whether it has broadcast the request, the set $\mathtt{recvSucc}$ of
participants whose yes vote it has tallied, whether it has tallied a no vote, and whether it has
decided and committed; each participant records whether it has voted, and whether it voted yes.
A participant vote ($\mathtt{Success}$/$\mathtt{Failure}$) and its coordinator-side tally
($\mathtt{RecvSuccess}$/$\mathtt{RecvFailure}$) are \emph{separate} actions: the vote injects a message that a later
tally consumes into $\mathtt{recvSucc}$ or $\mathtt{recvFail}$.
The safety property is \emph{agreement}: if the coordinator committed, every participant voted yes;
and if it decided without committing, some participant voted no.

\begin{figure}[t]
\small
\hrule\smallskip
\begin{minipage}[t]{0.46\linewidth}
\raggedright
\textbf{Protocol: Two-Phase Commit}

\smallskip
\textit{Node Types}
\begin{itemize}[leftmargin=1.2em,nosep]
    \item $\mathtt{Coord}(c)$: single coordinator.
    \item $\mathtt{Part}(p)$: participants.
\end{itemize}

\smallskip
\textit{Local State}
\begin{itemize}[leftmargin=1.2em,nosep]
    \item \emph{Coord} $c$:
    \begin{itemize}[leftmargin=1.2em,nosep]
        \item $\mathtt{hasPrepared}, \mathtt{hasDecided}, \mathtt{committed}: \mathsf{Bool}$
        \item $\mathtt{recvFail}: \mathsf{Bool}$
        \item $\mathtt{recvSucc}: \mathsf{Set}\ \mathtt{Part}$
    \end{itemize}
    \item \emph{Part} $p$:
    \begin{itemize}[leftmargin=1.2em,nosep]
        \item $\mathtt{hasVoted}, \mathtt{votedYes}: \mathsf{Bool}$
    \end{itemize}
\end{itemize}

\smallskip
\textit{Global State}
\begin{itemize}[leftmargin=1.2em,nosep]
    \item $\mathtt{sent}, \mathtt{recv}$: message buffers.
\end{itemize}

\smallskip
\textit{Messages}
\begin{itemize}[leftmargin=1.2em,nosep]
    \item $\mathtt{Prepare}(c{\to}p)$, $\mathtt{Success}(p{\to}c)$, $\mathtt{Failure}(p{\to}c)$,
    $\mathtt{Commit}(c{\to}p)$, $\mathtt{Abort}(c{\to}p)$.
\end{itemize}

\smallskip
\textit{Init}
\begin{itemize}[leftmargin=1.2em,nosep]
    \item all flags $\mathsf{False}$;\ $\mathtt{recvSucc} = \emptyset$; buffers empty.
\end{itemize}
\end{minipage}
\hfill\vrule\hfill
\begin{minipage}[t]{0.50\linewidth}
\raggedright
\textit{Action}
\begin{itemize}[leftmargin=1.2em,nosep]
    \item $\mathtt{Prepare}$ by Coord $c$:\\
    \hspace*{1em}\textit{require:}\ $\neg\, c.\mathtt{hasPrepared}$\\
    \hspace*{1em}broadcast $\mathtt{Prepare}(c, p)$ to every $p$\\
    \hspace*{1em}$c.\mathtt{hasPrepared} \leftarrow \mathsf{True}$

    \smallskip
    \item $\mathtt{Success}$ by Part $p$ on $m$:\\
    \hspace*{1em}\textit{require:}\ $m = \mathtt{Prepare}(c, p) \in \mathtt{sent}$\\
    \hspace*{2em}$\wedge\, m \notin \mathtt{recv} \,\wedge\, \neg\, p.\mathtt{hasVoted}$\\
    \hspace*{1em}add $m$ to $\mathtt{recv}$\\
    \hspace*{1em}send $\mathtt{Success}(p, c)$\\
    \hspace*{1em}$(p.\mathtt{hasVoted}, p.\mathtt{votedYes}) \leftarrow (\mathsf{True}, \mathsf{True})$\\
    \hspace*{1em}($\mathtt{Failure}$: sends $\mathtt{Failure}(p, c)$, sets only $\mathtt{hasVoted}$)

    \smallskip
    \item $\mathtt{RecvSuccess}$ by Coord $c$ on $m$:\\
    \hspace*{1em}\textit{require:}\ $m = \mathtt{Success}(p, c) \in \mathtt{sent} \,\wedge\, m \notin \mathtt{recv}$\\
    \hspace*{1em}add $m$ to $\mathtt{recv}$;\ $\mathtt{recvSucc} \mathrel{+}= p$\\
    \hspace*{1em}($\mathtt{RecvFailure}$: on $\mathtt{Failure}(p, c)$, sets $\mathtt{recvFail} \leftarrow \mathsf{True}$)

    \smallskip
    \item $\mathtt{Commit}$ by Coord $c$:\\
    \hspace*{1em}\textit{require:}\ $(\forall p.\ p \in \mathtt{recvSucc})$\\
    \hspace*{2em}$\wedge\, \neg\, \mathtt{recvFail} \,\wedge\, \neg\, \mathtt{hasDecided}$\\
    \hspace*{1em}broadcast $\mathtt{Commit}$\\
    \hspace*{1em}$(\mathtt{hasDecided}, \mathtt{committed}) \leftarrow (\mathsf{True}, \mathsf{True})$

    \smallskip
    \item $\mathtt{Abort}$ by Coord $c$:\\
    \hspace*{1em}\textit{require:}\ $\mathtt{recvFail} \,\wedge\, \neg\, \mathtt{hasDecided}$\\
    \hspace*{1em}broadcast $\mathtt{Abort}$\\
    \hspace*{1em}$\mathtt{hasDecided} \leftarrow \mathsf{True}$
\end{itemize}
\end{minipage}
\smallskip\hrule
\caption{Two-phase commit protocol, in the form of \defref{protocol}.}
\figlabel{tpc-protocol}
\end{figure}

\myparagraph{Scenario}
The developer's picture of a run is a single ordered story of three phases, recorded by
$\mathtt{scenario}_{\mathtt{2PC}}$; the decision guards read the coordinator's local state, so no ghost variables are needed:
\[
\begin{aligned}
    \mathtt{scenario}_{\mathtt{2PC}} &\triangleq \mathtt{PhaseChain}[P_1, P_2, P_3],\\
    P_1 &\triangleq \mathtt{Prepare},\\
    P_2 &\triangleq \mathtt{For}\bigl(p \in [\mathtt{min},\mathtt{max}]\bigr)\ \mathtt{body}(p),\\
    \mathtt{body}(p) &\triangleq \mathtt{PhaseChain}\bigl[\,\mathtt{Success}(p),\ \mathtt{RecvSuccess}(p)\,\bigr]\\
        &\quad\;\; +\ \mathtt{PhaseChain}\bigl[\,\mathtt{Failure}(p),\ \mathtt{RecvFailure}(p)\,\bigr],\\
    P_3 &\triangleq \assm(\mathtt{c}.\mathtt{recvFail}) \cdot \mathtt{Abort}\\
        &\quad\;\; +\ \assm(\forall q.\ q\in\mathtt{c}.\mathtt{recvSucc}) \cdot \mathtt{Commit},
\end{aligned}
\]
where $\mathtt{RecvSuccess}(p)$ and $\mathtt{RecvFailure}(p)$ tally $p$'s vote message.
The loop body $\mathtt{body}(p)$ is itself a choice of two phase-chains: participant $p$ either
takes the \emph{success} sub-phase --- vote $\mathtt{Success}$, then optionally the coordinator
tallies it --- or the symmetric \emph{failure} sub-phase.
Because each sub-phase is a $\mathtt{PhaseChain}$, the body denotes four run shapes per
participant: the empty run, the vote alone, and the vote followed by its tally, in each of the two
outcomes.

The representative processes participants strictly in order and tallies each vote immediately.
A real run interleaves distinct participants' votes and tallies freely, and these reorderings are
absorbed by commutativity.
Reading off the tabulated relation (\figref{tpc-theta}), $\theta$ declares independent: two votes
of distinct participants; a vote and a tally of a \emph{different} participant; two
tallies of different senders; and a vote or a success-tally sliding past $\mathtt{Abort}$ (which
reads only $\mathtt{recvFail}$).
It pointedly does \emph{not} commute a failure tally with a decision: that tally sets
$\mathtt{recvFail}$, enabling $\mathtt{Abort}$ and disabling $\mathtt{Commit}$, and this dependency
is the causal backbone the phase structure sequences.

\begin{figure}[t]
\centering
\begin{tabular}{r|cccc}
\multicolumn{1}{r}{} & $\mathtt{Vote}(p_2)$ & $\mathtt{Recv}(m_2)$ & $\mathtt{Commit}$ & $\mathtt{Abort}$ \\
\hline
$\mathtt{Vote}(p_1)$ & $p_1 \ne p_2$ & $p_1 \ne m_2.\sender$ & \tF & \tT \\
$\mathtt{RecvSuccess}(m_1)$ & $p_2 \ne m_1.\sender$ & $m_1.\sender \ne m_2.\sender$ & \tF & \tT \\
$\mathtt{RecvFailure}(m_1)$ & $p_2 \ne m_1.\sender$ & $m_1.\sender \ne m_2.\sender$ & \tF & \tF \\
$\mathtt{Commit}$       & \tF & \tF & \tF & \tF \\
$\mathtt{Abort}$        & \tT & \tT & \tF & \tF \\
\end{tabular}
\caption{Commutativity relation $\theta$ for two-phase commit, read row-then-column:
each cell gives the condition under which the ordered pair is in $\theta$ (\tT{} always, \tF{} never).
$\mathtt{Vote}$ stands for $\mathtt{Success}$ or $\mathtt{Failure}$, and $\mathtt{Recv}$ for
$\mathtt{RecvSuccess}$ or $\mathtt{RecvFailure}$, which behave identically except that a
$\mathtt{RecvFailure}$ sets $\mathtt{recvFail}$, the causal backbone $\mathtt{Abort}$ depends on.}
\figlabel{tpc-theta}
\end{figure}

\myparagraph{Coverage}
By \lemref{closedness-reduction}, coverage reduces to
$\mathtt{scenario}_{\mathtt{2PC}} \cdot a \inctop \mathtt{scenario}_{\mathtt{2PC}}$ for each action
$a$, discharged by rewriting the left-hand side to the right in the calculus of \secref{coverage}.
We write $S \triangleq \mathtt{scenario}_{\mathtt{2PC}}$.
Actions fired by the wrong role, and votes or tallies whose message the protocol never produced,
have unsatisfiable guards and are pruned at once by \rulename{Disabled}.
We present the broadcast $\mathtt{Prepare}$, a participant vote $\mathtt{Success}(p, m)$, the tally
$\mathtt{RecvSuccess}(m)$, and the decision $\mathtt{Commit}$; $\mathtt{Failure}$, $\mathtt{RecvFailure}$, and $\mathtt{Abort}$
are the symmetric mirrors.
The vote and the tally exhibit, as ring leader election did, the identification of a concrete
message with a scenario term.

\emph{Appending $\mathtt{Prepare}$.}
$\mathtt{Prepare}$ is recognized as already performed rather than commuted.
Expanding the phase-chain notation $\mathtt{PhaseChain}[P_1,P_2,P_3] = \mathbf 1 + P_1 + P_1 P_2 + P_1 P_2 P_3$
and distributing over the prefixes,
\[
\begin{aligned}
    S \cdot \mathtt{Prepare}
    &\equiv_\Sigma\; \bigl(\mathbf 1 + P_1 + P_1 P_2 + P_1 P_2 P_3\bigr)\cdot\mathtt{Prepare} && (\text{expand } \mathtt{PhaseChain})\\
    &\equiv_\Sigma\; \textstyle\sum_{\pi\in\{\mathbf 1, P_1, P_1P_2, P_1P_2P_3\}} \pi\cdot\mathtt{Prepare} && ({\rulename{Distrib}})
\end{aligned}
\]
On the three prefixes of positive length the coordinator has prepared, so
$\mathtt{c}.\mathtt{hasPrepared}$ holds and the appended $\mathtt{Prepare}$ is disabled;
\rulename{Disabled} rewrites each such summand to $\mathbf 0$.
On the empty prefix the appended $\mathtt{Prepare}$ is exactly the scenario's own $P_1$, matched back into $S$ by \rulename{Embed}.

\emph{Appending $\mathtt{Success}(p, m)$.}
The subtlety is the same as ring leader election's $\mathtt{Recv}$: the appended action names a
\emph{concrete} message $m$, whereas the scenario's vote term $\mathtt{Success}(p)$ computes its
message from state; we must identify the two before absorbing.
First, feasibility forces $m = \mathtt{Prepare}(\mathtt{c}, p)$ --- the vote is enabled only on the
prepare message addressed to $p$ --- and any other $m$ is pruned by \rulename{Disabled}.
Assume $m = \mathtt{Prepare}(\mathtt{c}, p)$.
Distributing over the four prefixes $\{\mathbf 1, P_1, P_1 P_2, P_1 P_2 P_3\}$: on $\mathbf 1$ no
$\mathtt{Prepare}$ has been broadcast, so the vote's message is not in $\mathtt{sent}$ and
\rulename{Disabled} prunes; the substantive prefix is $P_1 P_2$, treated next, and $P_1 P_2 P_3$
reduces to it, treated last.
On $P_1 P_2$ we peel the loop $P_2$ at $p$, writing
$L_{<p}, L_{>p}$ for the iterations below and above, and commute the appended vote leftward across
the tail:
\[
\begin{aligned}
    P_1 P_2 \cdot \mathtt{Success}(p, m)
    &\inctop\; P_1\cdot L_{<p}\cdot\mathtt{body}(p)\cdot L_{>p}\cdot\mathtt{Success}(p, m) && ({\rulename{Unfold-For}})\\
    &\inctop\; P_1\cdot L_{<p}\cdot\mathtt{body}(p)\cdot\mathtt{Success}(p, m)\cdot L_{>p} && ({\rulename{Commute}})
\end{aligned}
\]
sound because every action of $L_{>p}$ concerns a participant $q\ne p$, hence is
$\theta$-independent of $\mathtt{Success}(p,m)$ (the vote/vote and vote/tally cells of $\theta$
require $q \ne p$).
Now the key step: under the entry condition $\mathtt{cond} \equiv \{v_p = p\}$ pinning the loop
variable to $p$, the concrete $\mathtt{Success}(p, m)$ equals the scenario's parametric vote term,
whose message selector computes exactly $\mathtt{Prepare}(\mathtt{c}, p)$; \rulename{Action-Equiv}
replaces it,
\[
    \cdots\cdot\mathtt{body}(p)\cdot\mathtt{Success}(p, m)\cdot L_{>p}
    \;\inctop\;
    \cdots\cdot\mathtt{body}(p)\cdot\mathtt{Success}(p)\cdot L_{>p}
    \qquad ({\rulename{Action-Equiv}}).
\]
It remains to absorb $\mathtt{Success}(p)$ into $\mathtt{body}(p)$.
Expanding $\mathtt{body}(p)$ into its two sub-phase-chains and their prefixes gives six branches.
The branch $\mathbf 1$ opening the success sub-phase matches the appended $\mathtt{Success}(p)$ by
\rulename{Embed} (it becomes the vote of that sub-phase); on every other branch $p$ has already
voted, so $\mathtt{hasVoted}$ disables the appended vote and \rulename{Disabled} prunes it.

The remaining prefix $P_1 P_2 P_3$ ends in a decision.
After $\mathtt{Commit}$ every participant has voted, so the appended vote is disabled and \rulename{Disabled} prunes.
$\mathtt{Abort}$ does not require every participant to have voted, so a vote may still follow it;
the $(\mathtt{Abort}, \mathtt{Vote})$ cell is \tT{}, and the vote does not change $\mathtt{c}.\mathtt{recvFail}$, which the guard of $\mathtt{Abort}$ reads.
Writing $A \triangleq \assm(\mathtt{c}.\mathtt{recvFail}) \cdot \mathtt{Abort}$ for the abort summand of $P_3$,
\rulename{Commute} thus slides the appended vote leftward past it:
\[
    P_1 P_2 \cdot A \cdot \mathtt{Success}(p, m)
    \;\inctop\;
    P_1 P_2 \cdot \mathtt{Success}(p, m) \cdot A
    \qquad ({\rulename{Commute}}).
\]
This returns the appended vote to the $P_1 P_2$ situation already treated, with the decision trailing unchanged, so this prefix rewrites back into $S$ as well.
The four prefixes together give $S\cdot\mathtt{Success}(p,m)\inctop S$.

\emph{Appending $\mathtt{RecvSuccess}(m)$.}
The tally is where the coordinator consumes a vote message.
Feasibility forces $m = \mathtt{Success}(p, \mathtt{c})$ for some $p$; write
$b \triangleq \mathtt{RecvSuccess}(m)$.
As with the vote, we peel $P_2$ at $p = m.\sender$ and commute $b$ leftward across $L_{>p}$, whose
actions concern participants $q \ne p$ and so are $\theta$-independent of $b$ (the tally/vote and
tally/tally cells of $\theta$ require $q \ne p$, resp.\ distinct senders):
\[
\begin{aligned}
    P_1 P_2 \cdot b
    &\inctop\; P_1\cdot L_{<p}\cdot\mathtt{body}(p)\cdot L_{>p}\cdot b && ({\rulename{Unfold-For}})\\
    &\inctop\; P_1\cdot L_{<p}\cdot\mathtt{body}(p)\cdot b\cdot L_{>p} && ({\rulename{Commute}})
\end{aligned}
\]
Under $\mathtt{cond}$, \rulename{Action-Equiv} replaces the concrete $b$ by the scenario's term
$\mathtt{RecvSuccess}(p)$, as both tally the message $\mathtt{Success}(p, \mathtt{c})$.
Absorbing into $\mathtt{body}(p)$: in the success sub-phase-chain,
the tally is enabled exactly on the one-phase prefix $\mathtt{Success}(p)$,
where it completes the two-phase prefix and is matched by \rulename{Embed}.
On the empty prefix $p$ has not voted, so the tallied message is not yet sent and \rulename{Disabled}
prunes; on the completed prefix $m\in g.\mathtt{recv}$ and \rulename{Disabled} prunes; in the failure
sub-phase-chain $p$ sent no $\mathtt{Success}$ and \rulename{Disabled} prunes.
The run rewrites back into $S$, giving $S\cdot\mathtt{RecvSuccess}(m)\inctop S$.

\emph{Appending $\mathtt{Commit}$.}
$\mathtt{Commit}$ is a coordinator action, $\theta$-independent of the loop $P_2$, so it aligns
with the decision phase $P_3$ rather than commuting inward.
Only the $P_1 P_2$ and $P_1 P_2 P_3$ prefixes admit a decision; on the latter $\mathtt{hasDecided}$
disables it, pruned by \rulename{Disabled}.
On $P_1 P_2$ we split on the coordinator's yes-set with \rulename{Cond-Intro}.
When it is full,
\[
\begin{aligned}
    P_1 P_2 \cdot \mathtt{Commit}
    &\inctop\; P_1 P_2 \cdot \assm(\forall q.\ q\in\mathtt{c}.\mathtt{recvSucc}) \cdot \mathtt{Commit} && ({\rulename{Cond-Intro}})\\
    &\inctop\; P_1 P_2 P_3 \;\inctop\; S && ({\rulename{Embed}}),
\end{aligned}
\]
matching the commit summand of $P_3$.
When it is not full, $\mathtt{Commit}$'s own guard fails and \rulename{Disabled} prunes.
$\mathtt{Abort}$ is symmetric, splitting on $\mathtt{c}.\mathtt{recvFail}$ against the abort summand of
$P_3$; note that a failure tally does \emph{not} commute past $\mathtt{Abort}$ (the $(\mathtt{RecvFailure}, \mathtt{Abort})$
cell of $\theta$ is \tF{}), so the decision is only ever aligned, never reordered
past the tally that enables it.

\myparagraph{Correctness}
By \lemref{correctness-reduction}, correctness is checked on the representatives.
The $P_2$ loop carries the invariant
\[
\begin{aligned}
    J \;\equiv\;&
    \neg\,\mathtt{c}.\mathtt{hasDecided}
    \,\wedge\,
    \neg\,\mathtt{c}.\mathtt{committed}
    \,\wedge\,
    (\forall q \in \mathtt{c}.\mathtt{recvSucc}.\ q.\mathtt{votedYes})\\
    &\wedge\,
    (\mathtt{c}.\mathtt{recvFail} \Rightarrow \exists q.\ q.\mathtt{hasVoted} \wedge \neg\, q.\mathtt{votedYes}),
\end{aligned}
\]
together with the facts that each sent vote message matches how its sender voted.
Each branch of $\mathtt{body}(p)$ preserves $J$ under $\mathrm{wp}$, since a participant votes at most once (its $\mathtt{hasVoted}$ guard).
Before the decision, both conjuncts of agreement hold vacuously.
On the commit summand of $P_3$ the test forces $\forall q.\ q \in \mathtt{c}.\mathtt{recvSucc}$, so by $J$ every participant voted yes.
On the abort summand the test forces $\mathtt{c}.\mathtt{recvFail}$, so by $J$ some participant voted no, and $\mathtt{committed}$ stays false.

\subsection{Decentralized Lock}
\seclabel{cs-dlock}

The decentralized lock guarantees mutual exclusion: a single lock token
is passed among the nodes, and at most one node ever holds it.
It is instructive because it needs \emph{no} commutativity at all --- the protocol's own
invariant already forces every run into one canonical shape --- so it isolates the feasibility
and correctness reasoning from the commutativity that dominated the previous examples.

\myparagraph{Protocol}
Following \secref{prelim}, \figref{dlock-protocol} gives the protocol in full: a single node type;
one grant message type; and two actions.
An abstract schedule $\mathtt{holder}(t)$ fixes the intended holder at epoch $t$; the proof
is parametric in it, covering every order of handoffs at once.
A holder of epoch $t$ releases by sending a grant stamped $t+1$ to $\mathtt{holder}(t+1)$,
who takes the lock by consuming it.
The freshness guard $m.t > i.\mathtt{epoch}$ prevents a stale grant from reacquiring the lock.
The safety property is \emph{mutual exclusion}: at most one node holds the lock at any time.

\begin{figure}[t]
\small
\hrule\smallskip
\begin{minipage}[t]{0.48\linewidth}
\raggedright
\textbf{Protocol: Decentralized Lock}

\smallskip
\textit{Node Type}
\begin{itemize}[leftmargin=1.2em,nosep]
    \item $\mathtt{LockNode}(i)$: a single node type.
\end{itemize}

\smallskip
\textit{Parameter}
\begin{itemize}[leftmargin=1.2em,nosep]
    \item $\mathtt{holder}(t)$: the intended holder at epoch $t$.
\end{itemize}

\smallskip
\textit{Local State} (of Node $i$)
\begin{itemize}[leftmargin=1.2em,nosep]
    \item $\mathtt{hasLock}: \mathsf{Bool}$.
    \item $\mathtt{epoch}: \mathbb{N}$, epoch of the last grant taken.
\end{itemize}

\smallskip
\textit{Global State}
\begin{itemize}[leftmargin=1.2em,nosep]
    \item $\mathtt{sent}, \mathtt{recv}$: message buffers.
\end{itemize}

\smallskip
\textit{Message}
\begin{itemize}[leftmargin=1.2em,nosep]
    \item $\mathtt{grant}(i, t)$: to node $i$, timestamp $t$.
\end{itemize}

\smallskip
\textit{Init}
\begin{itemize}[leftmargin=1.2em,nosep]
    \item $\mathtt{holder}(0).\mathtt{hasLock} = \mathsf{True}$, others $\mathsf{False}$;\ all $\mathtt{epoch} = 0$; buffers empty.
\end{itemize}
\end{minipage}
\hfill\vrule\hfill
\begin{minipage}[t]{0.48\linewidth}
\textit{Action}
\begin{itemize}[leftmargin=1.2em,nosep]
    \item $\mathtt{Release}$ by Node $i$:\\
    \hspace*{1em}\textit{require:}\ $i.\mathtt{hasLock}$\\
    \hspace*{1em}send $\mathtt{grant}(\mathtt{holder}(t), t)$ with $t = i.\mathtt{epoch} + 1$\\
    \hspace*{1em}$i.\mathtt{hasLock} \leftarrow \mathsf{False}$

    \smallskip
    \item $\mathtt{Take}$ by Node $i$ on a grant $m$:\\
    \hspace*{1em}\textit{require:}\ $m.\mathtt{receiver} = i \,\wedge\, m \in \mathtt{sent}$\\
    \hspace*{2em}$\wedge\, m \notin \mathtt{recv} \,\wedge\, m.t > i.\mathtt{epoch}$\\
    \hspace*{1em}add $m$ to $\mathtt{recv}$\\
    \hspace*{1em}$i.\mathtt{hasLock} \leftarrow \mathsf{True}$\\
    \hspace*{1em}$i.\mathtt{epoch} \leftarrow m.t$
\end{itemize}
\end{minipage}
\smallskip\hrule
\caption{Decentralized lock protocol, in the form of \defref{protocol}.}
\figlabel{dlock-protocol}
\end{figure}

\myparagraph{Scenario}
The intuitive run is a single relay: the lock starts with $\mathtt{holder}(0)$ and is passed
along the chain of scheduled holders, each releasing so the next may take.
The scenario $\mathtt{scenario}_{\mathtt{Lock}}$ records this handoff chain.
A single scenario variable $c$ tracks the current epoch, so that at any point along the
chain the current (or most recent) holder is $\mathtt{holder}(c)$.
The chain is built from one \emph{hop}, which relays the lock once: the current holder
$\mathtt{holder}(c)$ releases, the next scheduled holder $\mathtt{holder}(c+1)$ takes the grant just sent, and $c$ advances.
The scenario initializes $c$, iterates hops, and admits a trailing lone release:
\[
\begin{aligned}
    \mathtt{scenario}_{\mathtt{Lock}}
    &\triangleq (c \leftarrow 0) \cdot \mathtt{PhaseChain}[\,\mathtt{Hop}^{*},\ \mathtt{Release}(\mathtt{holder}(c))\,],\\
    \mathtt{Hop}
    &\triangleq \mathtt{Release}(\mathtt{holder}(c)) \cdot \mathtt{Take}(\mathtt{holder}(c+1),\, m_{c+1}) \cdot (c \leftarrow c+1),
\end{aligned}
\]
where $m_t \triangleq \mathtt{grant}(\mathtt{holder}(t), t)$ is the grant stamped $t$ addressed to
$\mathtt{holder}(t)$.
Unfolding the $\mathtt{PhaseChain}$, the scenario denotes exactly the
runs $(c\leftarrow 0)\cdot\mathtt{Hop}^{k}$, optionally followed by a $\mathtt{Release}$.
The trailing optional release covers the one intermediate state in which the previous holder has
released but no one has yet taken --- a state with no holder, still safe.
It carries no increment, which the coverage argument for $\mathtt{Take}$ below supplies.

The striking feature is that this representative is essentially the \emph{only} feasible shape.
Because at most one node holds the lock, two releases can never occur back to back --- the second
releaser would need a lock it does not hold --- so every feasible run already factors as an
alternation of releases and takes, exactly the $\mathtt{Hop}^{k}$ and $\mathtt{Hop}^{k}\cdot\mathtt{Release}$
that the scenario denotes.
No scheduling freedom is left for commutativity to absorb, and accordingly the semi-commutation
$\theta$ is \emph{empty} (\figref{dlock-theta}): no pair of actions is declared independent.
This makes the example a clean contrast with the previous two: the same method that leaned
heavily on commutativity there needs none here, because the single-holder invariant does the
canonicalizing that commutativity would otherwise perform.

\begin{figure}[t]
\centering
\begin{tabular}{r|cc}
\multicolumn{1}{r}{} & $\mathtt{Release}$ & $\mathtt{Take}$ \\
\hline
$\mathtt{Release}$ & \tF & \tF \\
$\mathtt{Take}$    & \tF & \tF \\
\end{tabular}
\caption{Commutativity relation $\theta$ for the decentralized lock: empty. Every cell is \tF{}
--- no pair commutes --- because the single-holder invariant already forces one canonical run
shape, leaving nothing for commutativity to absorb.}
\figlabel{dlock-theta}
\end{figure}

\myparagraph{Coverage}
By \lemref{closedness-reduction}, coverage reduces to $\mathtt{scenario}_{\mathtt{Lock}} \cdot a
\inctop \mathtt{scenario}_{\mathtt{Lock}}$ for the two actions.
We write $S \triangleq \mathtt{scenario}_{\mathtt{Lock}}$ and abbreviate its two denoted run
families as $S_k \triangleq (c\leftarrow 0)\cdot\mathtt{Hop}^{k}$ (the run after $k$ completed hops,
lock held by $\mathtt{holder}(c)$) and $S_k^{R} \triangleq S_k \cdot \mathtt{Release}(\mathtt{holder}(c))$
(the same run followed by the trailing release, lock in flight, no holder).
Since $\theta$ is empty, the commutativity rules never fire; every case is settled by the
structure rules and the feasibility rule \rulename{Disabled}, driven by the loop invariant of $\mathtt{Hop}^{*}$
that tracks the lock token along the chain:
\[
    J \;\equiv\;
    \bigl(\forall n.\ n.\mathtt{hasLock} \Rightarrow n = \mathtt{holder}(c)\bigr)
    \,\wedge\,
    \mathtt{holder}(c).\mathtt{epoch} = c
    \,\wedge\,
    g.\mathtt{sent} \subseteq g.\mathtt{recv}.
\]
Reading the conjuncts: any holder is the scheduled $\mathtt{holder}(c)$, which carries epoch $c$, and no grant is in flight.
So $J$ holds after $S_k$, and after $S_k^{R}$ the only grant in flight is $m_{c+1}$.

\emph{Appending $\mathtt{Release}$.}
The scenario is a $\mathtt{PhaseChain}$ of two phases; expanding it and distributing the appended action over the
prefixes (abbreviating the leading assignment $(c\leftarrow 0)$ by $\mathtt{I}_0$):
\[
\begin{aligned}
    S \cdot \mathtt{Release}
    &\equiv_\Sigma\; \underbrace{\mathtt{I}_0\cdot\mathtt{Release}}_{(\mathrm i)}
      \;+\; \underbrace{\mathtt{I}_0\cdot\mathtt{Hop}^{*}\cdot\mathtt{Release}}_{(\mathrm{ii})} \\
    &\quad\;+\; \underbrace{\mathtt{I}_0\cdot\mathtt{Hop}^{*}\cdot\mathtt{Release}(\mathtt{holder}(c))\cdot\mathtt{Release}}_{(\mathrm{iii})}
      && (\text{expand } \mathtt{PhaseChain};\ {\rulename{Distrib}})
\end{aligned}
\]
Summands $(\mathrm i)$ and $(\mathrm{ii})$ are the runs ending in a held lock ($S_0$ and $S_k$);
we treat $(\mathrm{ii})$, of which $(\mathrm i)$ is the $k=0$ instance.
The appended $\mathtt{Release}$ names a concrete releasing node, which we first identify with the
scenario's term $\mathtt{Release}(\mathtt{holder}(c))$: since $\mathtt{Release}$ is enabled only at
a lock-holder, and $J$'s holder-uniqueness clause forces that holder --- after any number of hops
--- to be $\mathtt{holder}(c)$, the two resolve equally, and
\rulename{Action-Equiv} replaces one by the other.
The trailing phase of $\mathtt{PhaseChain}[\mathtt{Hop}^{*}, \mathtt{Release}(\mathtt{holder}(c))]$
is exactly $\mathtt{Hop}^{*}\cdot\mathtt{Release}(\mathtt{holder}(c))$, so \rulename{Embed} places
each summand back into $S$ (as $S_k^{R}$):
\[
    \mathtt{I}_0\cdot\mathtt{Hop}^{*}\cdot\mathtt{Release}
    \;\inctop\;
    \mathtt{I}_0\cdot\mathtt{Hop}^{*}\cdot\mathtt{Release}(\mathtt{holder}(c))
    \;\inctop\; S
    \qquad ({\rulename{Action-Equiv}},\ {\rulename{Embed}}).
\]
Summand $(\mathrm{iii})$ ends in a release, so no node holds the lock; by $J$'s holder-uniqueness
clause the appended $\mathtt{Release}$ is disabled and \rulename{Disabled} rewrites
$\mathtt{Release}(\mathtt{holder}(c))\cdot\mathtt{Release}$ to $\mathbf 0$, which \rulename{Annih}
absorbs and \rulename{Embed} discards.
The three summands together give $S\cdot\mathtt{Release}\inctop S$.

\emph{Appending $\mathtt{Take}(k, m)$.}
The interesting content of this case is that the only scenario runs able to accept a take are
those with the lock \emph{in flight}, and that supplying the missing increment turns the
appended take into a fresh hop.
We again unfold to the three phase prefixes and consider where an appended $\mathtt{Take}$ can be
feasible.
On the two prefixes ending in a held lock ($S_0$ and $S_k$), no grant is in flight.
The one live prefix is $S_k^{R}$, where the trailing release has put the grant $m_{c+1}$
in flight.
Writing $\mathtt{I}_0$ for $(c\leftarrow 0)$ and $R_c$ for $\mathtt{Release}(\mathtt{holder}(c))$:
\[
\begin{aligned}
    S_k^{R}\cdot\mathtt{Take}(k, m)
    &\inctop\; \mathtt{I}_0\cdot\mathtt{Hop}^{*}\cdot R_c\cdot\mathtt{Take}(\mathtt{holder}(c+1), m_{c+1}) && ({\rulename{Action-Equiv}})\\
    &\inctop\; \mathtt{I}_0\cdot\mathtt{Hop}^{*}\cdot R_c\cdot\mathtt{Take}(\mathtt{holder}(c+1), m_{c+1})\cdot(c\leftarrow c+1) && ({\rulename{Assign-Intro}})\\
    &\equiv_\Sigma\; \mathtt{I}_0\cdot\mathtt{Hop}^{*}\cdot\mathtt{Hop}
      \;\inctop\; \mathtt{I}_0\cdot\mathtt{Hop}^{*} \;\inctop\; S && (\text{fold }\mathtt{Hop};\ {\rulename{Embed}})
\end{aligned}
\]
\rulename{Action-Equiv} applies because $m_{c+1}$ is the only grant in flight, and the increment
appended by \rulename{Assign-Intro} is trivially taint-safe, as nothing follows it;
the last three terms are then precisely one hop $\mathtt{Hop}$.
Unfolding the star $\mathtt{Hop}^{*} \equiv_\Sigma \mathbf 1 + \mathtt{Hop}^{*}\cdot\mathtt{Hop}$,
the run $\mathtt{Hop}^{*}\cdot\mathtt{Hop}$ is embedded by \rulename{Embed} into $\mathtt{Hop}^{*}$,
so it lands back in $S$, giving $S\cdot\mathtt{Take}\inctop S$.

\myparagraph{Correctness}
By \lemref{correctness-reduction}, correctness is checked on the representatives, and it turns on
the same invariant $J$ that drove coverage.
Its holder-uniqueness conjunct $\forall n.\ n.\mathtt{hasLock}\Rightarrow n=\mathtt{holder}(c)$
plainly implies mutual exclusion, since any two holders are both $\mathtt{holder}(c)$ and
hence equal; the remaining conjuncts carry holder-uniqueness along the run.
Computing $\mathrm{wp}$ of the hop body shows each hop preserves $J$:
$\mathtt{holder}(c)$, of epoch $c$, drops its hold and sends $m_{c+1}$, the sole in-flight grant;
$\mathtt{holder}(c+1)$ takes it, becoming the only holder with epoch $c+1$ and emptying the in-flight set;
and $c \leftarrow c+1$ re-establishes $J$.
The optional trailing $\mathtt{Release}$ lies outside the loop: it drops the last holder,
leaving no holder at all, so mutual exclusion holds there too.
Correctness thus rests entirely on holder-uniqueness, with no commutativity argument in sight.

\subsection{Multi-Signature Payment (All)}
\seclabel{cs-msa}

Our next example is a multi-signature payment protocol, after Ivy's \texttt{multisig-all}:
a user submits payment requests, each request is proposed by its validator to every signer, and the
validator pays only once \emph{all} authorised signers have signed --- or cancels the request once it
has expired.
Unlike the rounds of Paxos, requests are not globally ordered: they are issued, signed, paid, and
expired concurrently.
We present the high-level ideas; the full model and proof are in the mechanization.

\myparagraph{Protocol}
\figref{msa-protocol} gives the protocol: validators, signers, and a single environment node
$\mathtt{Env}$ that issues requests and lets them expire; five message types; and seven actions.
Requests are numbered by a counter $\mathtt{reqCnt}$, at most $R$ of them, and request $c$ is routed
to the validator $\mathtt{cur\_validator}(c)$, a parameter of the proof.
A validator handles one request at a time ($\mathtt{holding}$), and only moves to a request of larger
id ($\mathtt{activeRound} \le c$).
The safety property is \emph{payment integrity}: every paid request was signed by every signer, each
of them authorised, and every cancelled request had expired.

\begin{figure}[t]
\small
\hrule\smallskip
\begin{minipage}[t]{0.42\linewidth}
\raggedright
\textbf{Protocol: Multi-Signature Payment (All)}

\smallskip
\textit{Node Types}
\begin{itemize}[leftmargin=1.2em,nosep]
    \item $\mathtt{Validator}(v)$, $\mathtt{Signer}(s)$, a single $\mathtt{Env}$.
\end{itemize}

\smallskip
\textit{Parameters}
\begin{itemize}[leftmargin=1.2em,nosep]
    \item $\mathtt{cur\_validator}(c)$: the validator of request $c$.
    \item $\mathtt{sig\_auth}(s)$: $s$ is authorised.
    \item $R$: the number of requests.
\end{itemize}

\smallskip
\textit{Local State} (of Validator $v$)
\begin{itemize}[leftmargin=1.2em,nosep]
    \item $\mathtt{holding}: \mathsf{Bool}$;\ $\mathtt{activeRound}: \mathbb{N}$
    \item $\mathtt{receivedSigs}: \mathsf{Set}\,(\mathbb{N} \times \mathsf{Signer})$
    \item $\mathtt{paidSet}, \mathtt{cancelledSet}: \mathsf{Set}\ \mathbb{N}$
\end{itemize}

\smallskip
\textit{Global State}
\begin{itemize}[leftmargin=1.2em,nosep]
    \item $\mathtt{reqCnt}: \mathbb{N}$;\ $\mathtt{expired}: \mathsf{Set}\ \mathbb{N}$
    \item $\mathtt{sent}, \mathtt{recv}$: message buffers
\end{itemize}

\smallskip
\textit{Messages} (each carrying a request id $c$)
\begin{itemize}[leftmargin=1.2em,nosep]
    \item $\mathtt{Request}(\mathtt{Env}{\to}v)$, $\mathtt{Timer}(\mathtt{Env}{\to}\mathtt{Env})$,
    $\mathtt{Proposal}(v{\to}s)$, $\mathtt{Sig}(s{\to}v)$, $\mathtt{Pending}(v{\to}v)$.
\end{itemize}

\smallskip
\textit{Init}
\begin{itemize}[leftmargin=1.2em,nosep]
    \item all $\mathtt{holding}$; $\mathtt{activeRound} = \mathtt{reqCnt} = 0$; sets and buffers empty.
\end{itemize}
\end{minipage}
\hfill\vrule\hfill
\begin{minipage}[t]{0.54\linewidth}
\raggedright
\textit{Action} (every receipt consumes its message)
\begin{itemize}[leftmargin=1.2em,nosep]
    \item $\mathtt{Request}$ by $\mathtt{Env}$:\
    \textit{require}\ $\mathtt{reqCnt} < R$;\\
    \hspace*{1em}send $\mathtt{Request}$ to $\mathtt{cur\_validator}(\mathtt{reqCnt})$ and a $\mathtt{Timer}$,\\
    \hspace*{1em}both with $c = \mathtt{reqCnt}$;\ $\mathtt{reqCnt} \mathrel{+}= 1$
    \item $\mathtt{Propose}$ by $v$ on $\mathtt{Request}(c)$:\\
    \hspace*{1em}\textit{require}\ $v.\mathtt{holding} \wedge v.\mathtt{activeRound} \le c$;\\
    \hspace*{1em}broadcast $\mathtt{Proposal}(c)$;\ send $\mathtt{Pending}(c)$ to $v$;\\
    \hspace*{1em}$v.\mathtt{holding} \leftarrow \mathsf{False}$;\ $v.\mathtt{activeRound} \leftarrow c$
    \item $\mathtt{Sign}$ by $s$ on $\mathtt{Proposal}(c)$ from $v$:\\
    \hspace*{1em}\textit{require}\ $\mathtt{sig\_auth}(s) \wedge c \notin \mathtt{expired}$;\ send $\mathtt{Sig}(c)$ to $v$
    \item $\mathtt{Receive}$ by $v$ on $\mathtt{Sig}(c)$ from $s$:\\
    \hspace*{1em}\textit{require}\ $\neg v.\mathtt{holding} \wedge v.\mathtt{activeRound} = c \wedge c \notin \mathtt{expired}$;\\
    \hspace*{1em}$v.\mathtt{receivedSigs} \mathrel{+}= (c, s)$
    \item $\mathtt{Pay}$ by $v$ on $\mathtt{Pending}(c)$:\\
    \hspace*{1em}\textit{require}\ $v.\mathtt{activeRound} = c \wedge c \notin \mathtt{expired}$\\
    \hspace*{2em}$\wedge\ \forall s.\ (c, s) \in v.\mathtt{receivedSigs}$;\\
    \hspace*{1em}$v.\mathtt{paidSet} \mathrel{+}= c$;\ $v.\mathtt{holding} \leftarrow \mathsf{True}$
    \item $\mathtt{Cancel}$ by $v$ on $\mathtt{Pending}(c)$:\\
    \hspace*{1em}\textit{require}\ $v.\mathtt{activeRound} = c \wedge c \in \mathtt{expired}$;\\
    \hspace*{1em}$v.\mathtt{cancelledSet} \mathrel{+}= c$;\ $v.\mathtt{holding} \leftarrow \mathsf{True}$
    \item $\mathtt{Expire}$ by $\mathtt{Env}$ on $\mathtt{Timer}(c)$:\ $\mathtt{expired} \mathrel{+}= c$
\end{itemize}
\end{minipage}
\smallskip\hrule
\caption{Multi-signature payment protocol (All), in the form of \defref{protocol}.}
\figlabel{msa-protocol}
\end{figure}

\myparagraph{Scenario}
The developer's picture of a run is one causal history per request: the request is proposed, signed
and collected signer by signer, and paid; independently, it may expire and then be cancelled.
The scenario issues every request up front and then runs the requests in ascending id order,
each as two phase chains:
\[
\begin{aligned}
    \mathtt{scenario}_{\mathtt{MSA}} &\triangleq \mathtt{Request}^{*} \cdot
        \mathtt{For}\bigl(r \in [0, R)\bigr)\ \mathtt{Happy}(r) \cdot \mathtt{Expiry}(r),\\
    \mathtt{Happy}(r) &\triangleq \mathtt{PhaseChain}\bigl[\,\mathtt{Propose}(r),\
        \mathtt{For}(s)\ \mathtt{PhaseChain}[\mathtt{Sign}(s, r), \mathtt{Receive}(s, r)],\
        \mathtt{Pay}(r)\,\bigr],\\
    \mathtt{Expiry}(r) &\triangleq \mathtt{PhaseChain}\bigl[\,\mathtt{Expire}(r),\ \mathtt{Cancel}(r)\,\bigr],
\end{aligned}
\]
where every term of request $r$ acts at $\mathtt{cur\_validator}(r)$ on request $r$'s message.
The two chains are sequenced rather than alternatives: an expiry may follow a payment, and it
disables every remaining action of the happy chain.

\myparagraph{Semi-commutativity}
The relation $\theta$ follows three principles, each justifying a reordering into the representative:
\begin{enumerate*}
    \item a $\mathtt{Request}$ commutes leftward past every action, as whatever ran before it concerns only older requests, so the requests can be issued up front;
    \item actions of distinct requests commute, except an expiry and the actions of its request that it disables; and
    \item at one validator, a larger request is proposed only after the smaller one is closed, after which $\mathtt{activeRound}$ disables every action of the smaller one, so the conditions stated on request ids hold vacuously there.
\end{enumerate*}

\myparagraph{Coverage}
Coverage follows the pattern of Paxos.
An appended $\mathtt{Request}$ commutes to the front and is absorbed by $\mathtt{Request}^{*}$.
Any other appended action belongs to some request $c < R$: we unfold the $\mathtt{For}$ at $c$, commute the
action leftward past the requests of larger id, and absorb it into the matching phase of
$\mathtt{Happy}(c)$ or $\mathtt{Expiry}(c)$, pruning the other prefixes with \rulename{Disabled}.
The side conditions rest on routing and bound invariants: each message travels between its request's
validator and the party it names, and carries an id below $R$.

\myparagraph{Correctness}
The loop invariant is the safety property together with $\mathtt{recordedAuth}$: every recorded
signature comes from an authorised signer.
Each step preserves both by its own guard or by that of the action just before it in the scenario:
$\mathtt{Pay}$ records a payment only when every signature is recorded; $\mathtt{Receive}$ records a
signature right after the $\mathtt{Sign}$ whose guard authorised its signer; and $\mathtt{Cancel}$
cancels only an expired request.
Notably, $\mathtt{recordedAuth}$ is preserved along the scenario but is not an inductive invariant of
the protocol in isolation; the scenario's order is what keeps it.

\subsection{Multi-Signature Payment (Majority)}
\seclabel{cs-msm}

The majority variant, after Ivy's \texttt{multisig-majority}, differs from \secref{cs-msa} in a single
place: a validator pays once a strict \emph{majority} of the $n_s$ signers has signed, rather than all
of them.
Writing $\mathtt{maj}_v(c) \triangleq 2\cdot\bigl|\{\, s : (c, s) \in v.\mathtt{receivedSigs} \,\}\bigr| > n_s$,
the signature conjunct of the $\mathtt{Pay}$ guard in \figref{msa-protocol} becomes $\mathtt{maj}_v(c)$.
The safety property weakens correspondingly: every paid request was signed by a strict majority of
signers, every recorded signer of it is authorised, and every cancelled request had expired.

The scenario, the semi-commutation, and the coverage proof are those of \secref{cs-msa}, unchanged.
The signer loop still visits every signer, and a run that pays on a majority is covered by the
prefix in which the remaining signers' phase chains are empty.
Correctness is the same argument, with $\mathtt{maj}_v(c)$ in place of the all-signed conjunct,
established at the payment step directly by the guard.
The two variants thus form a clean pair: the same scenario and coverage proof serve both, and only
the correctness check distinguishes them, illustrating the property-independence of coverage
(\remref{coverage-independent}).

\end{document}